\documentclass[12pt, draftclsnofoot, onecolumn]{IEEEtran}
\makeatletter
\newcommand*{\rom}[1]{\expandafter\@slowromancap\romannumeral #1@}
\makeatother
\IEEEoverridecommandlockouts
\ifCLASSINFOpdf
\else
\fi

\usepackage{float}
\usepackage{subfloat}
\usepackage[table]{xcolor}
\usepackage{graphics}
\usepackage{multicol}
\usepackage{lipsum}
\usepackage{color}
\usepackage[cmex10]{amsmath}
\usepackage{amsthm}
\usepackage{amssymb}
\usepackage{mathrsfs}
\usepackage{mathtools}
\usepackage{amsbsy}
\usepackage{xcolor}
\usepackage{mathrsfs}
\usepackage{setspace}
\usepackage{graphicx}
\usepackage{epstopdf}
\usepackage{lettrine}
\usepackage{psfrag}
\usepackage{newclude}
\usepackage[normalem]{ulem}
\usepackage{latexsym}
\usepackage{algpseudocode}
\usepackage{algorithm,algpseudocode}
\usepackage{algorithmicx}

\usepackage{multirow}
\usepackage{rotating}
\usepackage{booktabs}
\usepackage{amsmath}
\usepackage{wasysym}
\usepackage{mathrsfs}
\usepackage{bbm} 
\usepackage{filecontents}
\usepackage{amsmath,epsfig,cite,amsfonts,amssymb,psfrag,color}

\usepackage{apptools}

\usepackage{chngcntr}
\usepackage{optidef}
\ifCLASSOPTIONcompsoc
\usepackage[caption=false,font=normalsize,labelfon
t=sf,textfont=sf]{subfig}
\else
\usepackage[caption=false,font=footnotesize]{subfig}
\fi
\def\BibTeX{{\rm B\kern-.05em{\sc i\kern-.025em b}\kern-.08em
    T\kern-.1667em\lower.7ex\hbox{E}\kern-.125emX}}
\def\K{\mathcal{K}}
\def\L{\mathcal{L}}

\def\Ksum{\sum\limits_{k \in \mathcal{K}}}
\def\Lsum{\sum\limits_{l \in \mathcal{L}}}
\def\isum{\sum\limits_{i \in \mathcal{K}}}
\def\jsum{\sum\limits_{j \in \mathcal{L}}}
\def\hs{{\boldsymbol{h}_{k,r}^{s}}^H}
\def\hsss{{\boldsymbol{h}_{k,r}^{s}}}
\def\gs{\boldsymbol{g}_{l,r}^{s}}
\def\pr{\mathcal{P}}
\def\uk{u_{1,k}}
\def\ul{\boldsymbol{u}_l}
\def\h{\boldsymbol{h}}
\def\g{\boldsymbol{g}}
\def\flk{f_{l,k}}
\def\hhat{\hat{\boldsymbol{h}}}
\def\Hhat{\hat{\boldsymbol{H}}}
\def\ghat{\hat{\boldsymbol{g}}}
\def\TH{\hat{\boldsymbol{\Theta}}}
\def\hbH{\Bar{\boldsymbol{h}}^H}
\def\hb{\Bar{\boldsymbol{h}}}
\def\fb{\Bar{{f}}_{l,k}}
\def\gb{\Bar{\boldsymbol{g}}}
\def\w{\boldsymbol{w}}
\def\NDL{\sigma_{\text{DL}}^2}
\def\NUL{\sigma_{\text{UL}}^2}
\def\Nhat{\hat{\sigma}^2}
\def\kp{k^{\prime}}
\def\lp{l^{\prime}}
\def\Tw{\boldsymbol{T}}
\def\Htilda{\Tilde{\boldsymbol{H}}_k}
\def\lam{\boldsymbol{\Lambda}}
\def\lp{\lambda_l}
\def\bphi{\boldsymbol{\Phi}}
\def\THPhi{\hat{\boldsymbol{\Theta}}(\boldsymbol{\Phi})}
\def\THPhiGDA{\hat{\boldsymbol{\Theta}}(\boldsymbol{\Phi}^{(s)})}
\def\vHphi{\boldsymbol{v}^H(\boldsymbol{\Phi})}
\def\vphi{\boldsymbol{v}(\boldsymbol{\Phi})}
\def\aki{\boldsymbol{a}_{k,i}}
\def\blk{\boldsymbol{b}_{l,k}}
\def\cki{c_{k,i}}
\def\fonelk{f1_{l,k}}

\def\RSI{\text{RSI}(\boldsymbol{u}_l)}
\def\zlj{\boldsymbol{z}_{l,j}}
\def\dlj{d_{l,j}}
\def\dtil{\Tilde{\boldsymbol{d}}_u}
\def\etil{\Tilde{\boldsymbol{e}}}

\def\ytil{\Tilde{\boldsymbol{y}}}

\def\ztilnt{[\Tilde{\boldsymbol{z}}]_{nt}}
\def\ztilut{[\Tilde{\boldsymbol{z}}]_{ut}}

\def\stil{\Tilde{\boldsymbol{s}}}

\def\mtilnt{[\Tilde{\boldsymbol{m}}]_{nt}}
\def\mtilut{[\Tilde{\boldsymbol{m}}]_{ut}}

\def\xtilnt{[\Tilde{\boldsymbol{x}}]_{nt}}
\def\xtilut{[\Tilde{\boldsymbol{x}}]_{ut}}

\def\atilnt{[\Tilde{\boldsymbol{a}}]_{nt}}
\def\atilut{[\Tilde{\boldsymbol{a}}]_{ut}}

\def\btilnt{[\Tilde{\boldsymbol{b}}]_{nt}}
\def\btilut{[\Tilde{\boldsymbol{b}}]_{ut}}

\def\extmT{e^{j\phi_t}}

\def\exupT{e^{-j\phi_u}}

\def\exnpT{e^{-j\phi_n}}

\def\extm{e^{-j\phi_t}}

\def\exup{e^{j\phi_u}}

\def\exnp{e^{j\phi_n}}
\newcommand{\normm}[1]{\left\lVert#1\right\rVert}
\def\B{\mathcal{{B}}}
\def\C{\mathcal{{C}}}
\def\Q{\mathcal{{Q}}}
\def\T{\mathcal{{T}}}
\def\Bt{\Tilde{\mathcal{{B}}}}

\def\betkDL{\beta_k^{\text{DL}}}

\def\betkpDL{\beta_{k^{\prime}}^{\text{DL}}}
\def\betlUL{\beta_l^{\text{UL}}}

\def\betlpUL{\beta_{l^{\prime}}^{\text{UL}}}
\def\keciUEDL{\xi_{\text{UE}}^{\text{DL}}}
\def\keciUEUL{\xi_{\text{UE}}^{\text{UL}}}
\def\kecitUL{\xi_{\text{BS}}^{\text{UL}}}
\def\kecitDL{\xi_{\text{BS}}^{\text{DL}}}
\def\zDL{\boldsymbol{z}_{\text{BS}}^{\text{DL}}}
\def\zuserUL{z_{\text{UE}}^{\text{UL}}}
\def\zuserDL{{z_{{\text{UE}}}^{\text{DL}}}_k}
\def\zUL{\boldsymbol{z}_{\text{BS}}^{\text{UL}}}

\newcommand{\abss}[1]{{\lvert{#1}\rvert}^2}

\newtheorem{lemma}{Lemma}
\begin{document}

\title{Weighted Sum-Rate Maximization for Multi-IRS-assisted Full-Duplex Systems with Hardware Impairments}


    
 \author{\IEEEauthorblockN{\normalsize
    Mohammad Amin Saeidi, Mohammad Javad Emadi, Hamed Masoumi,  Mohammad Robat Mili, Derrick Wing Kwan Ng, \textcolor{black}{\textit{Senior Member, IEEE,}} and Ioannis Krikidis}, \textcolor{black}{\textit{Fellow, IEEE}}
 \thanks{M. A. Saeidi, M. J. Emadi and H. Masoumi are with the Electrical Engineering Department, Amirkabir University of Technology, Tehran, Iran. (E-mails: \{amin.saeidi, mj.emadi,hamed\_masoomy\}@aut.ac.ir). D. W. K. Ng is with the School of Electrical Engineering and Telecommunications, the University of New South Wales, Australia. (E-mail: w.k.ng@unsw.edu.au). M. Robat Mili and I. Krikidis are with the Department of Electrical and Computer Engineering, University of Cyprus, 1678 Nicosia, Cyprus (E-mails: Mohammad.Robatmili@ieee.org,  krikidis@ucy.ac.cy).}}
 
\maketitle

\begin{abstract}
Smart and reconfigurable wireless communication environments can be established by exploiting well-designed intelligent reflecting surfaces (IRSs) to shape the communication channels. In this paper, we investigate how multiple IRSs affect the performance of multi-user full-duplex communication systems under hardware impairment at each node, wherein the base station (BS) and the uplink users are subject to maximum transmission power constraints. Firstly, the uplink-downlink system weighted sum-rate (SWSR) is derived which serves as a system performance metric. Then, we formulate the resource allocation design for the maximization of SWSR as an optimization problem which jointly optimizes the beamforming and the combining vectors at the BS, the transmit powers of the uplink users, and the phase shifts of multiple IRSs. Since the SWSR optimization problem is non-convex, an efficient iterative alternating approach is proposed to obtain a suboptimal solution for the design problem considered and its complexity is also discussed. In particular, we firstly reformulate the main problem into an equivalent weighted minimum mean-square-error form and then transform it into several convex sub-problems which can be analytically solved for given phase shifts. Then, the IRSs phases are optimized via a gradient ascent-based algorithm. Finally, numerical results are presented to clarify how multiple IRSs enhance the performance metric under hardware impairment.
\end{abstract}

\section{Introduction}
To meet the required demands and to support the potential use cases in the fifth and sixth generations of wireless networks, e.g. the Internet-of-everything and the tactile internet, key enabling wireless technologies, in particular, massive multiple-input multiple-output (mMIMO), cell-free mMIMO, ultra-dense and device-to-device networks, higher frequency (millimeter-wave, terahertz) communications, drone-based communications, and the integration of terrestrial and satellite wireless networks has been proposed to enrich the network and support various use cases 
\cite{cellfreeMarzetta,KwanCellFree,rajatheva2020white,masoumi2019performance,ngo2017cell, boccardi2014five, masoumi2020cell}. Besides, the intelligent reflecting surface (IRS) is recently proposed to not only customize the propagation environment in wireless channels, but also it can be adopted as a complement solution to reduce the deployment costs of active antennas \textcolor{black}{used in conventional MIMO setups} \cite{di2019smart}. In practice, an IRS is fabricated as a thin metasurface composed of reflecting and phase-controllable elements, where each of them can manipulate the phase of the incident signal so as to shape the channel conditions \cite{rajatheva2020white}. By configuring the phase shifts introduced by the IRS, one can control the direction of the reflected signal towards a desired direction to enhance the received  signal-to-interference-plus-noise ratio (SINR) at some users \cite{OFDMmeets2020} or to improve the secrecy rate by \textcolor{black}{covering} the specific signals \cite{lu2020intelligentCovert,9163399,KwanSecure}. \textcolor{black}{Besides, the low-cost IRSs can be easily deployed on walls, buildings facades, road signs, etc., which makes it useful for various applications. e.g., smart-cities, homes, airports, and intelligent cars \cite{di2020smartSurvey}.}

Since the IRSs are deployed mainly with \textit{passive} elements, additive thermal noises and self-interference are generally negligible and imposing virtually no impact on the signal. Technically, by increasing the number of IRS elements, it has been shown that as IRS-based scheme can outperform the conventional \textit{active}  amplify/decode-and-forward relaying in both transferring energy and information \cite{emadiRelayVsIRS}, and provides higher energy-efficiency than the relay-assisted \textcolor{black}{systems} \cite{huang2019reconfigurable}. Moreover, as an emerging hardware technology, IRS is capable of providing a quadratic array gain compared to that of a linear array gain achieved by the conventional multi-antenna techniques \cite{bjornson2020power}.
Thus, IRS-aided communication is an interesting technique which serves as viable energy- and spectral-efficient solution to realize wireless communications and has a wide potential applications in emerging networks to improve performance of the system, e.g. IRS-assisted cell-free networks \cite{IRS-CF}, IRS-aided unmanned aerial vehicle communications \cite{ge2020joint,KwanUAV},  two-way IRS-assisted communications \cite{IRS-TwoWay},  employment of the IRS in wireless power transfer \cite{MIMOpowerIRS}, and integrating backscatter link with IRS \cite{zhao2020backscatter} are studied in the literature.

Due to the numerous potential practical applications of IRS deployments, IRS-assisted wireless communication systems have received increasing attention from academia \textcolor{black}{to investigate its fundamental limitations as well as enabling practical design.} Specifically, in  \cite{Multipoint}, an IRS-aided multicell wireless network was considered wherein joint processing coordinated multipoint transmission from multiple base stations (BSs) was conducted by exploiting an optimized IRS. In \cite{RuiZhang-Network}, single- and multi-user multiple-input single-output (MISO) IRS-aided systems were studied, while joint active and passive beamforming problem were designed to minimize the total transmit power at the BS by using the semidefinite relaxation \cite{boyd2004convex} and alternating optimization techniques. Besides, a downlink multi-group multicast communication system supported by an IRS was considered in \cite{Multicast}, and the sum-rate of all the multicast groups was maximized by optimizing the precoding matrix at the BS and the phase shifts at the IRS. Also, the authors in \cite{krikidisIRS2019lowcomplexity} proposed  low-complexity and energy-efficient schemes adopting a random phase  rotation at each element of the IRS to overcome high propagation losses drawback of the IRS-assisted communications.  Similar to the system model of  \cite{RuiZhang-Network}, the authors in \cite{zappone2020resource}  assumed channel-matched beamforming to maximize geometric mean of downlink SINR for all the users. Also, the gradient-based method was used for optimizing the reflection coefficients of the IRS.


Despite the fruitful results in the literature, the performance of wireless communication systems is mainly limited as the uplink and downlink are always separated orthogonally which underutilize the system resources. As a remedy, the full-duplex (FD) communications have been proposed which can almost double the spectral efficiency compared to the traditional half-duplex (HD) technology. In particular, FD transceivers are allowed to transmit in downlink and uplink simultaneously in the same frequency band at the cost of introducing strong self-interference \cite{FD1, KrikidisFD2018heterogeneous,FD_Secrecy}. To enable effective FD communication, resource allocations for a \textit{single} IRS-assisted cognitive networks was designed in \cite{schober2020resource} to maximize the sum rates of secondary network while controlling interference leakage on the primary users. \textcolor{black}{Hence, in our work, to investigate the impact of employing multiple IRSs on the FD communication systems, we propose to adopt multiple IRSs}. On the other hand, it is well-known that the performance of the communication systems can be heavily degraded if the hardware devices are not perfect due to the phase noise, sampling frequency offset, in-phase/quadrature-phase imbalance, quantization errors, non-linearity effect, etc. \textcolor{black}{In particular, if advanced complex signal processing and expensive high-quality devices are in use, non-negligible residual hardware impairment (HI) remains after calibration\cite{HI}.} Thus, fundamental performance of communication systems in the presence of HI has been analyzed in the literature from different perspectives \cite{HI_sec,HI_mMIMO, masoumi2019performance}. It is worth mentioning that the hardware impairment \textcolor{black}{alongside with employing multiple IRSs operating in a FD communication system} was \textit{not} considered in all the aforementioned works on IRS-assisted communications \textcolor{black}{\cite{IRS-CF,ge2020joint,KwanUAV,IRS-TwoWay,MIMOpowerIRS,zhao2020backscatter,Multipoint,RuiZhang-Network,Multicast,krikidisIRS2019lowcomplexity,zappone2020resource,schober2020resource}}, i.e. the transceivers of the legitimate users and the base station are equipped with perfect hardware components. Thus, analyzing  FD IRS-assisted systems with imperfect devices is of highly interest.

In this paper, to realize cost- and performance-efficient multi-user systems, a FD IRS-assisted system is studied to improve system performance, while assuming \emph{imperfect} transceivers to investigate how the FD IRS-assisted system behaves in the presence of HI. It is assumed that \emph{multiple} IRSs coexist in the network to cooperatively support the uplink (UL) and the downlink (DL) users while interacting with a multi-antenna BS. Since the UL and DL communications are performed in a FD manner, not only the signal of the UL users cause interference to the DL user, but also the BS is also subject to a non-negligible self-interference. Thus, we aim to design efficient resource allocation algorithms for maximizing the weighted system sum-rate (SWSR). In the following, the contributions of this paper are summarized.
\begin{itemize}
    \item To simultaneously support the UL and DL data transmissions, \textcolor{black}{a \textit{full-duplex} system is integrated with} \textit{multiple} IRSs to provide and enhance performance of the \textit{multi-user} communications between the UL users-to-BS and the DL users-to-BS pairs, respectively, in the presence of \textit{imperfect} transceivers. 
    
    \item The achievable rates of the uplink and the downlink users are derived, while the UL-DL weighted system sum-rate is formulated to be maximized. To maximize the considered performance metric,  we jointly optimize the beamforming vector for the downlink users subject to the maximum power constraint at the BS and the combining, i.e. data recovery, vector of the uplink users at the BS. Moreover, the UL power allocations are derived subject to the maximum power constraint at each UL user and the optimal phase shifts of IRSs' elements are derived. 
    
    \item Since the mentioned optimization problem is not jointly concave over the optimization parameters, \textcolor{black}{a suboptimal algorithm based on the iterative alternating optimization approach is designed}. Specifically, for a given IRSs' phase shift matrices, we reformulate the optimization problem into an equivalent weighted minimum mean-square-error (WMMSE) problem to obtain the DL beamformer, the UL combining vector, and the UL users' transmit powers, iteratively. Afterwards, for the given beamformer, the combiner and the power allocation solutions, we handle the challenging IRSs' phase shifts optimization problem via a gradient-based algorithm to obtain a suboptimal solution. 
    \item Our numerical results show that employing multiple IRSs can significantly enhance the SWSR performance compared with that of the conventional system without IRSs or fixed phase IRSs. Also, deploying multiple IRSs can effectively overcome the non-ideal hardware effects at both the users and the BS. Finally, it is shown that  deploying IRSs close to both of the UL and DL users, results in an evident improvement compared with the case that uses only a single IRS in the system.
\end{itemize}

\textit{Organization}:\ The rest of this paper is organized as follows. In Section \ref{sec:sysmodel}, the considered system model is introduced and the SWSR is derived. Section \ref{sec:Problem statement and analysis} formulates the problem and provides its analysis to determine the details of the proposed algorithm. Numerical results are discussed in Section \ref{sec:simulation}, and finally, Section \ref{sec:concolusion} concludes the paper.

\textit{Notations}:\ $\mathbb{C}^{M \times N}$ denotes the space of $M \times N$ complex valued matrices. $\mathbb{H}^{M}$ denotes the set of all complex Hermitian matrix with dimension $M$. For a square matrix $\boldsymbol{F}$, $\text{Tr}(\boldsymbol{F})$ denotes its trace and $\boldsymbol{F}\succeq \boldsymbol{0}$ denotes that $\boldsymbol{F}$ is positive semidefinite matrix. $\text{rank}(\boldsymbol{F})$ denotes the rank of $\boldsymbol{F}$. For complex-valued vector $\boldsymbol{x}$, $\lvert{\boldsymbol{x}\rvert}$ denotes its Euclidean norm. For complex-valued scalar $x$, $\Re(x)$ and $\Im(x)$, denote the real part and imaginary part of $x$, respectively. $x^*$ stand for the conjugate of $x$ and $\Vec{\boldsymbol{0}}$ denotes a zero vector. The matrix $\boldsymbol{I}_N$ represents a $N\times N$ identity matrix. For independent and identically distributed (i.i.d.) random variable (RV) ${s}$, ${s}\sim\mathcal{CN}\left(0,\sigma\right)$ denotes that the RV has complex Gaussian distribution with zero mean and variance $\sigma$.
\begin{figure}[!t]
   \centering
   \psfrag{bs}[][][0.85]{FD-BS}
   \psfrag{hself}[][][1]{\ \ $\boldsymbol{H}^{\text{SI}}$}
   \psfrag{Hr}[][][1]{\! \! \! \!$\boldsymbol{H}_r$}
   \psfrag{HR}[][][1]{\!\!\!\!\!\!$\boldsymbol{H}_R$}
   \psfrag{irsr}[][][0.85]{${\ \ \ \ \ \ \text{IRS}}_r$}
   \psfrag{irsR}[][][0.85]{${\text{IRS}}_R$}
   \psfrag{hk}[][][1]{\ \ $\h_k$}
   \psfrag{gl}[][][1]{\! \! $\g_l$}
   \psfrag{hkr}[][][1]{\ \ $\boldsymbol{h}_{k,r}^{s}$}
   \psfrag{hkR}[][][1]{$\boldsymbol{h}_{k,R}^{s}$}
   \psfrag{glr}[][][1]{ \ \ \ $\boldsymbol{g}_{l,r}^{s}$}
   \psfrag{glR}[][][1]{$\boldsymbol{g}_{l,R}^{s}$}
   \psfrag{flk}[][][1]{$\flk$}
   \psfrag{downlink users}[][][0.85]{Downlink Users}
   \psfrag{uplink users}[][][0.85]{Uplink Users}
\includegraphics[scale=0.4]{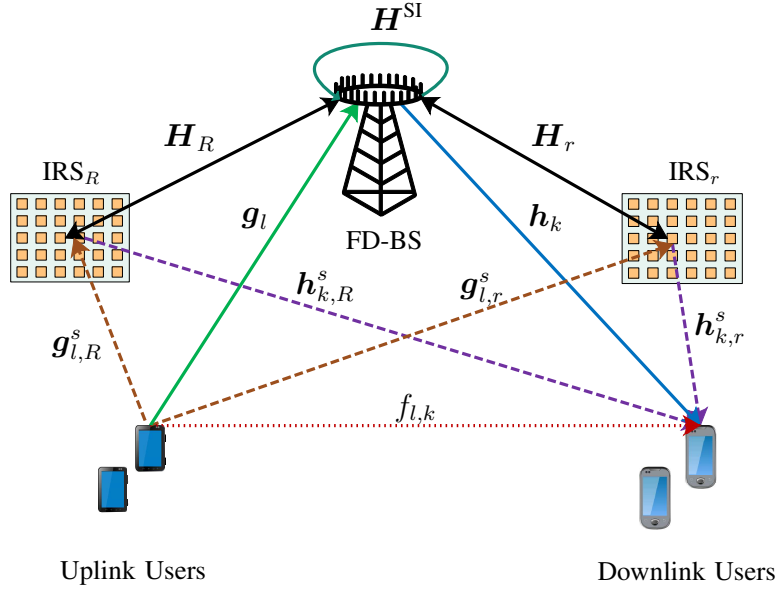}
\caption{A FD multi-IRS aided multi-user system.}
\label{fig:sysmodel}
\end{figure}
\section{System Model}
\label{sec:sysmodel}
As depicted in Fig. \ref{fig:sysmodel}, we consider a multi-IRS aided multi-user FD system consisting of one $N_t$-antenna BS, ${K}$ single-antenna downlink users, ${L}$ single-antenna uplink users, \textcolor{black}{and ${R}$ IRSs in which the number of elements for the $r$-th IRS is $M_r$.} We denote the sets of DL users, UL users, IRSs, and elements of the $r$-th IRS  as $\mathcal{K}=\{1,...,K\}$, $\mathcal{L}=\{1,...,L\}$, $\mathcal{R}=\{1,...,R\}$, and $\mathcal{M}_{r}=\{1,...,M_r\}$, respectively. Since the UL users operate at the same frequency as the DL users, the UL signals interfere with the DL users. Also, the IRSs reflect all the incident signals received simultaneously from BS and UL users. \textcolor{black}{Furthermore, by sending training pilots, the BS can estimate all the channel coefficients  \cite{ChannelEstimation}, and thus we assume that perfect channel state information (CSI) is available at the BS for resource allocation design.} In the following, we firstly introduce the hardware impairment model and then present the signal transmission and reception at different nodes. Subsequently, the UL and DL achievable rates and the SWSR are presented.
\subsection{Hardware Impairment Model}
In practice, the non-idealness of hardware introduces noisy distortions to the transmitted/received signal \cite{masoumi2019performance}. In general, this effect can be modeled by $x_{d}=\sqrt{\xi} x+z$, where the input signal to the non-ideal hardware is denoted by $x$, and $\xi \in[0,1]$ indicates the hardware quality factor. In the sequel, we use $\keciUEDL$, $\keciUEUL$, $\kecitDL$, $\kecitUL$ to represent the hardware quality factors of DL users, UL users, the BS transmitter, and the BS receiver, respectively. Also, the distortion is modeled by $z \sim \mathcal{C} \mathcal{N}\left(0,\left(1-\xi\right) \mathbb{E}\left\{|x|^{2}\right\}\right)$, which is independent from the input signal  $x$.
\subsection{Signal Transmissions and Receptions}
Signal model of each node is discussed in the following. The FD-BS transmits the super-imposed precoded signal $\boldsymbol{x}^{\text{DL}}=\sqrt{\kecitDL}(\sum\limits_{k\in\mathcal{K}}\boldsymbol{w}_k s_k)+\zDL$ to the $K$ DL users, where $s_{k}\sim\mathcal{CN}\left(0,1\right)$ and $\boldsymbol{w}_k \in \mathbb{C}^{N_t \times 1}$ denote the i.i.d. information symbol for the $k$-th DL user and the corresponding BS transmit beamforming, respectively, and $\zDL\sim\mathcal{CN}\left(0,\bar{\kecitDL}\Ksum\abss{\w_k} \boldsymbol{I}_{N_t}\right)$ denotes the distortion caused by hardware impairment at the BS transmitter where $\bar{\kecitDL}=(1-\kecitDL)$. The $l$-th UL user transmits  ${x}_{l}^{\text{UL}}=\sqrt{\keciUEUL}\sqrt{\rho_l}q_l + \zuserUL$, where $q_{l}\sim\mathcal{CN}\left(0,1\right)$ is i.i.d. information symbol and $\rho_l$ indicates the transmit power of the $l$-th UL user, and $\zuserUL\sim\mathcal{CN}\left(0,\bar{\keciUEUL}\rho_{l}\right)$ denotes the distortion caused by hardware impairment at the $l$-th UL user where $\bar{\keciUEUL}=(1-\keciUEUL)$. 

By neglecting multiple reflected signals from each IRS, and assuming that the delay among multiple paths introduced by the $R$ IRSs is negligible compared to the symbol duration, the $k$-th DL user receives the following signal
\begin{equation*}
     y_k^{\text{DL}}\!= \underbrace{\sqrt{\keciUEDL\kecitDL}\!\left(\!\boldsymbol{h}_k^H\! + \sum\limits_{r\in\mathcal{R}}\! \hs \boldsymbol{\Theta}_r \boldsymbol{H}_r\!\right)\! \w_k s_k}_{\text{Desired signal}}  \!+\underbrace{\sqrt{\keciUEDL\kecitDL}\!\left(\!\boldsymbol{h}_k^H\! + \sum\limits_{r\in\mathcal{R}}\! \hs \boldsymbol{\Theta}_r \boldsymbol{H}_r\!\right)\! \sum\limits_{i\neq k}^K\w_i s_i}_{\text{Multi-user interference}}
     \end{equation*}
     \begin{equation}\label{eqn: DL-Received}
     + \sqrt{\keciUEDL}\!\left(\!\boldsymbol{h}_k^H\! + \sum\limits_{r\in\mathcal{R}}\! \hs \boldsymbol{\Theta}_r \boldsymbol{H}_r\!\right)\!\zDL\!+ \underbrace{\sqrt{\keciUEDL}\sum\limits_{l\in\mathcal{L}}\! \left(\!{f}_{l,k}\! + \sum\limits_{r\in\mathcal{R}}\! \hs \boldsymbol{\Theta}_r \gs\!\right)\! {x}_{l}^{\text{UL}}}_{\text{UL interference signals and their reflections from IRSs}}\! +\zuserDL\!+ n^{\text{DL}},
\end{equation}
where $\boldsymbol{h}_k \in \mathbb{C}^{N_t \times 1}$, $\hsss \in \mathbb{C}^{M_{r} \times 1}$, and $\boldsymbol{H}_r \in \mathbb{C}^{M_r \times N_t}$ denote the channels between the BS and the $k$-th DL user, the channels between the $r$-th IRS and the $k$-th DL user, and the channel matrix between the BS and the $r$-th IRS, respectively, and ${f}_{l,k} \in \mathbb{C}$, $\gs\in \mathbb{C}^{M_{r} \times 1}$ represent the channels between the $l$-th UL user and the $k$-th DL user and the channels between the $l$-th UL user and the $r$-th IRS, respectively. Besides, diagonal matrix $\boldsymbol{\Theta}_r=\operatorname{diag}\left(e^{j \phi_{r,1}}, e^{j \phi_{r,2}}, \ldots, e^{j \phi_{r,M_r}}\right)$ expresses  the phase shift matrix of the $r$-th IRS while, $\phi_{r,i} \in [0,2\pi), \forall i\in\mathcal{M}_r$, is the phase shift applied to the incident signal via the $r$-th IRS, and $\zuserDL\sim\mathcal{CN}\left(0,\sigma_{\zuserDL}^2\right)$ denotes the distortion caused by hardware impairment at the $k$-th DL user. Moreover, the distortion variance at the $k$-th DL user is derived as
\begin{equation}\label{eqn:zuserDLVar}
    {\sigma_{\zuserDL}^2}=\bar{\keciUEDL}\left(\kecitDL\isum\abss{\hbH_k\w_i}+\bar{\kecitDL}\abss{\hbH_k}\isum\abss{\w_i}+\Lsum\abss{\fb}\rho_l\right),
\end{equation}
where $\hbH_k =\h_k^H+\hhat_k\TH\Hhat $, $\fb=\flk+\hhat_k\TH \ghat_l$, and $\hat{\boldsymbol{h}}_k\hat{\boldsymbol{\Theta}}\hat{\boldsymbol{H}} = \sum\limits_{r\in\mathcal{R}} \hs \boldsymbol{\Theta}_r \boldsymbol{H}_r$, $\hat{\boldsymbol{h}}_k\hat{\boldsymbol{\Theta}}\hat{\boldsymbol{g}}_l=\sum\limits_{r\in\mathcal{R}} \hs \boldsymbol{\Theta}_r \gs$, $\hat{\boldsymbol{h}}_k=\left[\boldsymbol{h}_{k,1}^s,\ldots,\boldsymbol{h}_{k,R}^s\right]^T$, $\hat{\boldsymbol{g}}_l=\left[\boldsymbol{g}_{l,1}^s,\ldots,\boldsymbol{g}_{l,R}^s\right]^T$, $\hat{\boldsymbol{H}}=\left[\boldsymbol{H}_1,\ldots,\boldsymbol{H}_R\right]^T$. The diagonal matrix $\hat{\boldsymbol{\Theta}} = \operatorname{diag}\left(\boldsymbol{\Theta}_1,...,\boldsymbol{\Theta}_R\right) \in \mathbb{H}^{M}$ is a block matrix such that its diagonal entries contain the phase shifts of the all $R$ IRSs and $M = M_1+\ldots+M_R$.
Also, $n^{\text{DL}}\sim\mathcal{CN}\left(0,\sigma_{\text{DL}}^2\right)$ models the circular symmetric complex additive white Gaussian noise (AWGN) at the DL users. 

The received signal at the BS is given by
\begin{equation}\label{eqn:BS Reception}
  \boldsymbol{y}^{\text{UL}}=\underbrace{\sqrt{\kecitUL}\sum\limits_{l\in\mathcal{L}}\left(\boldsymbol{g}_l + \sum\limits_{r\in\mathcal{R}} \boldsymbol{H}_{r}^H \boldsymbol{\Theta}_r \gs\right) {x}_{l}^{\text{UL}}}_{\text{Direct signals and their reflections from IRSs}} + \underbrace{\boldsymbol{H}^{\text{SI}}\boldsymbol{x}^{\text{DL}}}_{\text{Residual self interference}}+\zUL + \boldsymbol{n}^{\text{UL}}, 
\end{equation}
where $\boldsymbol{g}_l\in\mathbb{C}^{N_t\times1}$ is channel between the BS and the $l$-th UL user, and the term $\boldsymbol{H}^{\text{SI}}\boldsymbol{x}^{\text{DL}}$ indicates the residual self-interference (RSI) \cite{WMMSE}. Similar to \cite{SI-1}, we assume that $\boldsymbol{H}^{\text{SI}}$ is unknown at the BS and each element  has i.i.d. complex zero-mean Gaussian distribution with variance $\hat{\sigma}$, and $\boldsymbol{n}^{\text{UL}}\sim\mathcal{CN}\left(\boldsymbol{0},\sigma_{\text{UL}}^2\boldsymbol{I}_{N_t}\right)$ models AWGN at the BS. Also, $\zUL\sim\mathcal{CN}\left(0,\sigma_{\zUL}^2\boldsymbol{I}_{N_t}\right)$ denotes the distortion caused by hardware impairment at the receiver of the BS, and the distortion variance $\sigma_{\zUL}^2$ is derived as follows
\begin{equation}\label{eqn:zULVar}
    \sigma_{\zUL}^2=\bar{\kecitUL}\left(\jsum\abss{\gb_j}\rho_j+\hat{\sigma}^2\isum\abss{\w_i}(\kecitDL+\bar{\kecitDL}N_t)\right),
\end{equation}
where $\gb_l = \g_l + \Hhat^H \TH \ghat_l$ and $\bar{\kecitUL}=(1-\kecitUL)$.
\subsection{System Weighted Sum-Rate}
In the following, achievable rates of the DL and the UL are derived and the SWSR is presented. By using (\ref{eqn: DL-Received}), the achievable data rate in  \textcolor{black}{bits per channel use (bpcu)} of the $k$-th DL user becomes
\begin{equation}\label{eqn: DL rate}
    R_k^{\text{DL}}=\log_2(1+\gamma_k) ~~ \text{[bpcu]},
\end{equation}
where $\gamma_k$ is the DL SINR and is given by
\begin{equation}\label{eqn: DLSINR}
    \gamma_k=\frac{{\keciUEDL\kecitDL\lvert{\hbH_k\boldsymbol{w}_k}\rvert}^2 }{\kecitDL\sum\limits_{i\neq k}^K {\lvert{\hbH_k\boldsymbol{w}_i}\rvert}^2 +\bar{\keciUEDL}\kecitDL\abss{\hbH_k\w_k}+\bar{\kecitDL}\abss{\hbH_k}\isum\abss{\w_i}+ \sum\limits_{l\in\mathcal{L}}{\lvert{\fb}\rvert}^2 \rho_l + \sigma_{\text{DL}}^2}.
\end{equation}
 After receiving the signal (\ref{eqn:BS Reception}) at the BS, it applies the combining vector $\boldsymbol{u}_l \in \mathbb{C}^{N_t \times 1}$ to recover the data symbol of $l$-th UL user, that is  \textcolor{black}{$\hat{q}_l=\boldsymbol{u}_l^H \boldsymbol{y}^{\text{UL}}$.} Thus the achievable transmission rate of the $l$-th UL user becomes
\begin{equation}\label{eqn: UL rate}
    R_l^{\text{UL}}=\log_2(1+\gamma_l) ~~ \text{[bpcu]},
\end{equation}
where $\gamma_l$ is the UL SINR and is given by
\begin{equation}\label{eqn: ULSINR}
    \gamma_l=\frac{\keciUEUL\kecitUL{\lvert{\boldsymbol{u}_l^H\gb_l}\rvert}^2 \rho_l }{\kecitUL\sum\limits_{j\neq l}^K {\lvert{\boldsymbol{u}_l^H\gb_j}\rvert}^2 \rho_j +\bar{\keciUEUL}\kecitUL\abss{\ul^H\gb_l}\rho_l+\abss{\ul}\bar{\kecitUL}\jsum\abss{\gb_j}\rho_j+ \text{RSI}(\boldsymbol{u}_l) + \sigma_{\text{UL}}^2 {\lvert{\boldsymbol{u}_l} \rvert}^2}.
\end{equation}

Also, since $\boldsymbol{H}^{\text{SI}}$ is known to the BS, to simplify the effect of the residual self-interference, we use average RSI power similar to \cite{FDRSI}. Thus, the average RSI power at the BS for the $l$-th user is given by
\begin{equation}\label{eqn: RSI}
\begin{split}
    \text{RSI}(\boldsymbol{u}_l)&=\mathbb{E}\left\lbrace{\lvert{\boldsymbol{u}_l^H \boldsymbol{H}^{\text{SI}} \boldsymbol{x}^{\text{DL}} + \ul^H\zUL}\rvert}^2\right\rbrace \\ &={\hat{\sigma}}^2 {\lvert{\boldsymbol{u}_l} \rvert}^2  \sum\limits_{k\in\mathcal{K}}{\lvert \boldsymbol{w}_{k} \rvert}^2\left(\kecitUL+\kecitDL-\kecitUL\kecitDL+\bar{\kecitUL}\bar{\kecitDL}N_t\right).   
\end{split}
\end{equation}

Therefore, the SWSR is defined as
\begin{equation}\label{eqn: SSR}
    \text{{SWSR}}= \alpha_1 \sum\limits_{k\in \K} \beta_k R_k^{\text{DL}} + \alpha_2 \sum\limits_{l\in \L} \beta_l R_l^{\text{UL}},
\end{equation}
where $\betkDL\geq 0$ and $\betlUL\geq 0$ are constants which are introduced to control the priority of $k$-th DL user and $l$-th UL user, respectively, and $\alpha_1\geq 0$ and $\alpha_2\geq 0$ control weights of sum-rate at the DL and the UL, respectively. 
\section{Optimization Problem formulation and analysis}\label{sec:Problem statement and analysis}
To maximize the SWSR of the considered scenario, the following optimization problem is introduced
\begin{maxi!}[2]
	{\boldsymbol{w}_k,\boldsymbol{u}_l,\rho_l,\boldsymbol{\hat{\Theta}}}{\alpha_1 \sum\limits_{k\in \K} \betkDL R_k^{\text{DL}} + \alpha_2 \sum\limits_{l\in \L} \betlUL R_l^{\text{UL}}\label{eqn:p1}}
	{\label{problem:p1main}}{\mathcal{P}_1: \ \ }
	\addConstraint{\Ksum\ {\lvert{\boldsymbol{w}_k}\rvert}^2}{\leq P_{\text{max}}^{\text{BS}}\label{eqn:p1c1}}
	\addConstraint{\rho_l}{\leq P_{\text{max}}^l,\ \forall l\label{eqn:p1c2}}
	\addConstraint{0\leq \phi_{r,m}}{\leq 2\pi, \ \forall r,m,\label{eqn:p1c3}}
\end{maxi!}
$\!\! \! \! \text{where}$ (\ref{eqn:p1c1}) denotes the maximum power constraint at the BS with the maximum transmit power $P_{\text{max}}^{\text{BS}}$, (\ref{eqn:p1c2}) represents the maximum transmit power constraint of each UL user wherein  $P_{\text{max}}^l$ is the maximum transmit power at the $l$-th UL user, and (\ref{eqn:p1c3}) indicates the IRSs phase constraints. 

It is known that the optimization problem $\pr_1$ is non-convex and obtaining its globally optimal solution is challenging. As a compromise approach, we adopt an alternating optimization method which aims to achieve a suboptimal solution of the problem.
\textcolor{black}{Firstly, for a given phase shift matrices, the corresponding optimization problem is transformed into an equivalent WMMSE formulation which facilitate the development of an iterative method which converges to a stationary point of the corresponding objective function  with low computational complexity \cite{shi2011iteratively}. In the following, we decompose this equivalent optimization problem into a sequence of convex sub-problems, and the beamformer, the combining vector at the BS and the transmitted power of the UL users are optimized. Afterwards,  for the given solutions, we optimize the phase shift matrices via a gradient-based algorithm;} this process continues until the convergence. Finally, the complexity of the proposed algorithms is discussed.
\subsection{Equivalent WMMSE Optimization Problem for a Given $\boldsymbol{\hat{\Theta}}$}
For a given $\boldsymbol{\hat{\Theta}}$, by applying a similar WMMSE framework with the work in \cite{WMMSE-largeScale,masoumi2019performance}, the optimization problem $\pr_1$ is  transformed into the following equivalent WMMSE version
\begin{mini!}[2]
	{\substack{\boldsymbol{w}_k,\boldsymbol{u}_l,\uk\\ \rho_l,\mu_k^{\text{DL}},\mu_l^{\text{UL}}}}{\alpha_1 \Ksum \betkDL( \mu_k^{\text{DL}} e_k^{\text{DL}} - \ln{\mu_k^{\text{DL}}}) + \alpha_2 \Lsum \betlUL( \mu_l^{\text{UL}} e_l^{\text{UL}} - \ln{\mu_l^{\text{UL}}}) \label{eqn:p2}}
	{\label{problem:p2main}}{\mathcal{P}_2: \ \ }
	\addConstraint{\Ksum\ {\lvert{\boldsymbol{w}_k}\rvert}^2}{\leq P_{\text{max}}^{\text{BS}}\label{eqn:p2c1}}
	\addConstraint{\rho_l}{\leq P_{\text{max}}^l,\ \forall l,\label{eqn:p2c2}}
\end{mini!}
$\! \! \! \text{wherein}$ \textcolor{black}{$\mu_k^{\text{DL}}$ and $\mu_l^{\text{UL}}$ are weight factors for DL and UL, respectively.} Moreover, $e_k^{\text{DL}}$ and $e_l^{\text{UL}}$ are defined as
\begin{equation}\label{ek}
    \begin{split}
    e_k^{\text{DL}} = & \mathbb{E}\left\lbrace{\lvert{\hat{s}_k - s_k}\rvert^2}\right\rbrace =  \mathbb{E}\left\lbrace{\lvert{\uk y_k^{\text{DL}} - s_k}\rvert^2}\right\rbrace \\
    = & \lvert{\uk}\rvert^2 \left(\kecitDL\sum\limits_{i \in \K} \lvert{\hbH_k\w_i}\rvert^2 + \bar{\kecitDL} \abss{\hbH_k} \isum \abss{\w_i} + \Lsum\rho_l\lvert{\fb}\rvert^2 +\NDL \right) \\ -&2 \Re\left(\sqrt{\keciUEDL\kecitDL}\uk\hbH_k\w_k\right) + 1,
    \end{split}
\end{equation}
where $s_k$ is detected by the decoding coefficient $\uk \in \mathbb{C}$, i.e. $\hat{s}_k=\uk y_k^{\text{DL}}$, and
\begin{equation}\label{el}
    \begin{split}
        e_l^{\text{UL}} = & \mathbb{E}\left\lbrace{\lvert{\hat{q}_l - q_l}\rvert^2}\right\rbrace =  \mathbb{E}\left\lbrace{\lvert{\ul^H \boldsymbol{y}^{\text{UL}} - q_l}\rvert^2}\right\rbrace = \kecitUL\sum\limits_{j \in \L}\lvert{\ul^H\gb_j}\rvert^2\rho_j \\ +& \lvert{\ul}\rvert^2 \left( \bar{\kecitUL} \jsum \abss{\gb_j}\rho_j +  \Ksum\abss{\w_k}\Nhat\left(\kecitUL+\kecitDL-\kecitUL\kecitDL+\bar{\kecitUL}\bar{\kecitDL}N_t\right)+\NUL\right)
         \\ -&2 \Re\left(\sqrt{\keciUEUL\kecitUL}\ul^H\gb_l\sqrt{\rho_l}\right) + 1.
    \end{split}
\end{equation}
In the following, in order to derive the optimal values of $\{ \ul,\uk,\mu_k^{\text{DL}},\mu_l^{\text{UL}}\}$, $\w_k$ and $\rho_l$, we transform the problem $\pr_2$ into several sub-problems. \textcolor{black}{Although, $\pr_2$ is not a jointly convex problem,  for each of the variables $\boldsymbol{w}_k,\boldsymbol{u}_l,\uk, \rho_l,\mu_k^{\text{DL}}, \ \text{and} \  \mu_l^{\text{UL}}$, the problem is convex and the corresponding solution can be achieved. By exploiting this fact, we propose an alternating procedure to address the sub-problems of $\pr_2$ which is summarized in Algorithm 1 and is explained in the following.}
\subsubsection{Optimal Values of $\{\ul,\uk,\mu_k^{\text{DL}},\mu_l^{\text{UL}}\}$}
For a given set of ${\lbrace \ul,\w_k,\mu_k^{\text{DL}},\mu_l^{\text{UL}},\rho_l \rbrace}$, we first present the following optimization problem $\pr_{2.1}$ to find optimal value of $\uk$
\begin{mini}
	{\uk}{\alpha_1 \Ksum \betkDL( \mu_k^{\text{DL}} e_k^{\text{DL}} - \ln{\mu_k^{\text{DL}}}). \label{eqn:p2.1}}
	{}{\mathcal{P}_{2.1}: \ \ }
\end{mini}
\textcolor{black}{Since the objective function \eqref{eqn:p2.1} is a convex function of $\uk$}, by taking the first derivative of (\ref{eqn:p2.1}) with respect to $\uk$ and set it equal to zero, we have 
\begin{equation}\label{uk}
    \uk^{\text{opt}} = \frac{\sqrt{\keciUEDL\kecitDL}\w_k^H \hb_k}{\kecitDL\sum\limits_{i \in \K}\abss{\hbH_k \w_i} +\bar{\kecitDL} \abss{\hbH_k} \isum \abss{\w_i} + \Lsum\rho_l\abss{\fb} + \NDL}.
\end{equation}
Similarly, for a given set of ${\lbrace \w_k,\uk,\mu_k^{\text{DL}},\mu_l^{\text{UL}},\rho_l \rbrace}$, problem $\pr_2$ is simplified as 
\begin{mini}
	{\ul}{\alpha_2 \Lsum \betlUL( \mu_l^{\text{UL}} e_l^{\text{UL}} - \ln{\mu_l^{\text{UL}}}). \label{eqn:p2.2}}
	{}{\mathcal{P}_{2.2}: \ \ }
\end{mini}
Thus, by computing the first derivative of (\ref{eqn:p2.2}) respect to $\ul$ and set it equal to zero, the optimal value of the combining vector at the BS is derived as  
\begin{equation}\label{ul}
\begin{split}
    \ul^{\text{opt}} =& \Bigg(\kecitUL\sum\limits_{j \in L}\rho_j\gb_j\gb_j^H + \bar{\kecitUL} \jsum \abss{\gb_j}\rho_j \\ +&\left(\Ksum\abss{\w_k}\Nhat \left(\kecitUL+\kecitDL-\kecitUL\kecitDL+\bar{\kecitUL}\bar{\kecitDL}N_t\right)+\NUL\right)\boldsymbol{I}_{N_t} \Bigg)^{-1} \sqrt{\keciUEUL\kecitUL}\sqrt{\rho_l}\gb_l.
\end{split}
\end{equation}
Finally, to find the optimal values of ${e_k^{\text{DL}}}$ and ${e_l^{\text{UL}}}$, as the objective function \eqref{eqn:p2} is convex with respect to ${e_k^{\text{DL}}}$ and ${e_l^{\text{UL}}}$, by taking the first derivative of (\ref{eqn:p2}) with respect to these parameters separately and then set them equal to zero, we have
\begin{subequations}
\begin{align}
{\mu_k^{\text{DL}}}^{\text{opt}} = {e_k^{\text{DL}}}^{-1},\label{muk} \\
{\mu_l^{\text{UL}}}^{\text{opt}} = {e_l^{\text{UL}}}^{-1}.\label{mul}
\end{align}
\end{subequations}
\subsubsection{Optimizing the BS Beamforming Vector}\label{sec:wOPT}
For a given set of ${\lbrace \ul,\uk,\mu_k^{\text{DL}},\mu_l^{\text{UL}},p_l \rbrace}$, optimization problem $\mathcal{P}_2$ is rewritten as
\begin{mini!}[2]
	{\boldsymbol{w}_k}{\alpha_1 \Ksum \betkDL(\mu_k^{\text{DL}} e_k^{\text{DL}} - \ln{\mu_k^{\text{DL}}}) + \alpha_2 \Lsum \betlUL(\mu_l^{\text{UL}} e_l^{\text{UL}} - \ln{\mu_l^{\text{UL}}}) \label{eqn:p2.3}}
	{\label{probelm:p2.3main}}{\mathcal{P}_{2.3}: \ \ }
	\addConstraint{\Ksum\ {\lvert{\boldsymbol{w}_k}\rvert}^2}{\leq P_{\text{max}}^{\text{BS}},\label{eqn:p2.3c1}}
\end{mini!}
$\! \! \! \! \! \text{where}$ ${p}_l=\sqrt{\rho_l}$. Since the objective function (\ref{eqn:p2.3}) and the constraint (\ref{eqn:p2.3c1}) are convex, the problem $\pr_{2.3}$ can be solved by a standard solver such as CVX \cite{cvx}. Nevertheless, to obtain more system design insight, we solve the problem $\pr_{2.3}$ through the Lagrangian method. The Lagrangian function is given by
\begin{equation}\label{eqn:LagrangianForwk}
    \Tilde{\mathcal{L}}_w(\w_k,\lambda)=\alpha_1 \Ksum \betkDL\mu_k^{\text{DL}} e_k^{\text{DL}} + \alpha_2 \Lsum \betlUL\mu_l^{\text{UL}} e_l^{\text{UL}} +\lambda\left(\Ksum\ {\lvert{\boldsymbol{w}_k}\rvert}^2- P_{\text{max}}^{\text{BS}}\right),
\end{equation}
where $\lambda \geq 0$ is the corresponding Lagrangian multiplier for the constraint (\ref{eqn:p2.3c1}). Then, in order to derive the optimal stationary point of $\w_k$, we take the first derivative of $\Tilde{\mathcal{L}}(\w_k,\lambda)$ with respect to $\w_k$, and set $\frac{\partial\Tilde{\mathcal{L}}_w}{\partial\w_k}=0$. Thus, we have  
\begin{equation}\label{eqn: wk optimal with lambda}
\begin{split}
     \w_k(\lambda) &= \Bigg(\alpha_1\left(\kecitDL\sum\limits_{k^{\prime}\in \K}\betkpDL\mu_{\kp}\abss{u_{1,\kp}}\hb_{\kp}\hb_{\kp}^H + \bar{\kecitDL} \sum\limits_{k^{\prime}\in \K}\betkpDL\mu_{\kp}\abss{u_{1,\kp}} \abss{\hb_{\kp}}\boldsymbol{I}_{N_t} \right) \\ &+ \alpha_2  \betlUL\Nhat \left(\kecitUL+\kecitDL-\kecitUL\kecitDL+\bar{\kecitUL}\bar{\kecitDL}N_t\right) \Lsum\mu_l^{\text{UL}}\abss{\ul}\boldsymbol{I}_{N_t} + \lambda\boldsymbol{I}_{N_t} \Bigg)^{-1} \\ & \ \ \ \ \ \ \ \ \ \ \ \ \ \ \ \ \ \   \times \alpha_1 \sqrt{\keciUEDL\kecitDL} \betkDL\mu_k^{\text{DL}}  \uk^* \hb_k.
\end{split}
\end{equation}
To obtain $\w_k$, one needs to determine optimal value of $\lambda$, as well. Due to the complementary slackness condition for the constraint (\ref{eqn:p2.3c1}) \cite{boyd2004convex}, we have
\begin{equation}\label{eqn: slackness 1}
    \lambda\left(\Ksum{\lvert{\boldsymbol{w}_k}\rvert}^2- P_{\text{max}}^{\text{BS}}\right) = 0.
\end{equation}
In the following, we present Lemma~\ref{lemma1} to derive optimal values of $\lambda$ and $\boldsymbol{w}_k$.
\begin{lemma}\label{lemma1}
$J(\lambda)=\Ksum\abss{\w_k(\lambda)}$ is a monotonically decreasing function of $\lambda$.
\end{lemma}
\begin{proof}
Let us define matrix $\boldsymbol{A} \in \mathbb{H}^{N_t}$ as follows
\begin{equation}\label{eqn:A}
\begin{split}
    \boldsymbol{A} =& \alpha_1\left(\kecitDL\sum\limits_{k^{\prime}\in \K}\betkpDL\mu_{\kp}\abss{u_{1,\kp}}\hb_{\kp}\hb_{\kp}^H +\bar{\kecitDL} \sum\limits_{k^{\prime}\in \K}\betkpDL\mu_{\kp}\abss{u_{1,\kp}} \abss{\hb_{\kp}}\boldsymbol{I}_{N_t}\right) \\ + & \alpha_2  \betlUL\Nhat\left(\kecitUL+\kecitDL-\kecitUL\kecitDL+\bar{\kecitUL}\bar{\kecitDL}N_t\right)\Lsum\mu_l^{\text{UL}}\abss{\ul}\boldsymbol{I}_{N_t} , 
\end{split}
\end{equation}
where $\boldsymbol{A} \succeq \boldsymbol{0}$, and assuming that its rank is $N_{\tau}$ such that $N_{\tau} \leq N_t$. Thus, the eigenvalue decomposition of $\boldsymbol{A}$ becomes
\begin{equation}\label{eqn: eigenDecompos of A}
    \boldsymbol{A} = [\Tw_1 \ \Tw_2]\operatorname{diag}(\lam_1 \ \lam_2) [\Tw_1 \ \Tw_2]^H,
\end{equation}
wherein the first $N_{\tau}$ eigenvectors corresponding to the $N_{\tau}$ strictly positive eigenvalues are denoted by $\Tw_1$, $\lam_2=\Vec{\boldsymbol{0}}$, and $\lam_1$ is a diagonal matrix of $N_{\tau}$ strictly positive eigenvalues of $\boldsymbol{A}$. Hence, we can write the matrix $\boldsymbol{A}$ as
\begin{equation}\label{eqn: Decomposition of A}
    \boldsymbol{A} = \Tw_1 \lam_1 \Tw_1^H .
\end{equation}
Now, by using (\ref{eqn: Decomposition of A}), 
 the constraint (\ref{eqn:p2.3c1}) is reformulated as
\begin{equation}\label{eqn: J(lambda)}
\begin{split}
    J(\lambda) =& \Ksum \abss{\w_k}=\Ksum\w_k^H\w_k \\ = &\Ksum\keciUEDL\kecitDL\abss{\alpha_1}\abss{\mu_k^{\text{DL}}}\abss{\uk} \text{Tr}\left(\Tw_1(\lam_1 + \lambda\boldsymbol{I}_{N_t} )^{-1}\Tw_1^H\hb_k\hb_k^H\Tw_1(\lam_1 + \lambda\boldsymbol{I}_{N_t} )^{-1}\Tw_1^H\right) \\
    = & \Ksum \keciUEDL\kecitDL \abss{\alpha_1}\abss{\mu_k^{\text{DL}}}\abss{\uk} \text{Tr}\left((\lam_1 + \lambda\boldsymbol{I}_{N_t} )^{-2} \Htilda\right) \\
    = & \Ksum \keciUEDL\kecitDL \abss{\alpha_1}\abss{\mu_k^{\text{DL}}}\abss{\uk}\sum\limits_{i=1}^{N_{\tau}} \frac{[\Htilda]_{i,i}}{(y_i + \lambda)^2},
\end{split}
\end{equation}
where $\Htilda=\Tw_1^H\hb_k\hb_k^H\Tw_1$. Further, $[\Htilda ]_{i,i}$ and $y_i$ denote the $i$-th diagonal entry of $\Htilda$ and $\lam_1$, respectively. It is observed that the function $J(\lambda)$ is monotonically decreasing function of $\lambda$. 
\end{proof}
According to Lemma~\ref{lemma1}, if $J(0) \leq P_{\text{max}}^{\text{BS}}$ then $\w_k^{\text{opt}}=\w_k(0)$ for all $k$, otherwise  $J(\lambda^{\text{opt}})=P_{\text{max}}^{\text{BS}}$ must be solved to find $\lambda^{\text{opt}}$.

Based on the monotonic characteristic of $J(\lambda)$, the optimal dual variable $\lambda^{\text{opt}}$ can be found by using the bi-section search method. Moreover, to shrink the search space of the bi-section method, we adopt the following upper bound 
\begin{equation}\label{eqn:lambda Max 1}
    J(\lambda) \leq \Ksum \keciUEDL\kecitDL \abss{\alpha_1}\abss{\mu_k^{\text{DL}}}\abss{\uk}\sum\limits_{i=1}^{N_{\tau}} \frac{[\Htilda]_{i,i}}{(\lambda_{\text{max}})^2} \overset{\Delta}{=} P_{\text{max}}^{\text{BS}}.
\end{equation}
Therefore, the upper bound of $\lambda$ is given by
\begin{equation}\label{eqn:lambda Max 2}
    \lambda_{\text{max}} = \sqrt{\frac{\Ksum \keciUEDL\kecitDL \abss{\alpha_1}\abss{\mu_k^{\text{DL}}}\abss{\uk}\sum\limits_{i=1}^{N_{\tau}} [\Htilda]_{i,i}}{P_{\text{max}}^{\text{BS}}}}.
\end{equation}
As a result, the optimal downlink beamformer vector is computed as
\begin{equation}\label{eqn:wOpt}
\w_k^{\text{opt}}=\left\{\begin{array}{ll}
\w_k(0), &  J(0) \leq P_{\text{max}}^{\text{BS}} \\
\left(\boldsymbol{A} + \lambda^{\text{opt}}\boldsymbol{I}_{N_t} \right)^{-1} \alpha_1 \sqrt{\keciUEDL\kecitDL}\ \betkDL\mu_k^{\text{DL}} \uk^* \hb_k, & \text { o/w }.
\end{array}\right.
\end{equation}
\textcolor{black}{According to Lemma~1, when $\lambda=0$, the function $J(\lambda)$ achieves its maximum value. Depends on the values of $\alpha_1, \alpha_2, \betkDL, \betlUL$ which are given arbitrary parameters, and the maximum transmit power at the UL users, the $J(0) $ is  less than or equal  $P_{\text{max}}^{\text{BS}}$. So, the BS transmit power equals to $J(0)$. On the other hand, \textcolor{black}{depends on the priority of the DL or UL sum-rat and the transmit power at the UL users, if the maximum value of the function $J(\lambda)$, i.e. $J(0)$,} is higher than BS maximum transmit power, since $J(\lambda)$ is a decreasing function of $\lambda$, by computing the $\lambda^{\text{opt}}$, the value of $J(\lambda^{\text{opt}})$ reduces and equals $P_{\text{max}}^{\text{BS}}$ to meet the maximum power transmit constraint at the BS.} \textcolor{black}{Intuitively, according to the priority of the UL/DL in SWSR, we can increase the maximum transmit power of the BS or that of the UL users. For instance, if the priority of the UL sum rates is higher than the DL sum rates, the transmit power at the BS cannot reaches its maximum value because it  degrades the SWSR. Conversely, if the DL sum rates is more important that that of the UL one, by raising the value of the transmit power of the UL users, the BS needs to increase its power to combat the impact of the stronger interference due to the UL users.}
\subsubsection{Optimal Power Transmission at Uplink}
For a given set of ${\lbrace \ul,\uk,\mu_k^{\text{DL}},\mu_l^{\text{UL}},\w_k \rbrace}$, to derive optimal transmission power of users at uplink, we have the following optimization problem $\pr_{2.4}$.
\begin{mini!}[2]
	{p_l}{\alpha_1 \Ksum \betkDL(\mu_k^{\text{DL}} e_k^{\text{DL}} - \ln{\mu_k^{\text{DL}}}) + \alpha_2 \Lsum \betlUL(\mu_l^{\text{UL}} e_l^{\text{UL}} - \ln{\mu_l^{\text{UL}}}) \label{eqn:p2.4}}
	{\label{problem:p2.4main}}{\mathcal{P}_{2.4}: \ \ }
	\addConstraint{p_l^2}{\leq P_{\text{max}}^l,\forall l.\label{eqn:p2.4c1}}
\end{mini!}
Similar to the optimization problem $\pr_{2.3}$, problem $\pr_{2.4}$ can be solved by standard convex problem solvers or by the Lagrangian method. Hence, the Lagrangian function of $\pr_{2.4}$ is given by
\begin{equation}\label{eqn:LagrangianFor pl}
    \Tilde{\mathcal{L}}_p(p_l,\lp)=\alpha_1 \Ksum \betkDL\mu_k^{\text{DL}} e_k^{\text{DL}} + \alpha_2 \Lsum \betlUL \mu_l^{\text{UL}} e_l^{\text{UL}} +\Lsum  \lp\left(p_l^2-P_{\text{max}}^l\right),
\end{equation}
where $\lp \geq 0,$ for all $l$ are the corresponding Lagrangian multipliers for the constraints given in (\ref{eqn:p2.4c1}). Now by setting the first derivative of (\ref{eqn:LagrangianFor pl}) to zero, i.e. $\frac{\partial\Tilde{\mathcal{L}}_p}{\partial p_l}=0$, the transmit power of UL users is obtained as
\begin{equation}\label{eqn:p opt}
    p_l = \frac{\alpha_2 \sqrt{\keciUEUL\kecitUL} \betlUL\mu_l^{\text{UL}}\Re(\ul^H\gb_l)}{\alpha_1 \Ksum \betkDL\abss{\fb}\mu_k^{\text{DL}}\abss{\uk} + \alpha_2 \kecitUL\sum\limits_{l^{\prime} \in \L} \betlpUL \mu_{l^{\prime}} \abss{\gb_{l}^H \boldsymbol{u}_{l^{\prime}}}+ \bar{\kecitUL}\abss{\gb_l}\jsum\mu_{j^{\prime}} \abss{\boldsymbol{u}_j} + \lp}.
\end{equation}
On the other hand, the Lagrangian multiplier $\lp$ must satisfy the following complementary slackness condition
\begin{equation}\label{eqn:complemmentary s for p}
    \lp\left(p_l^2-P_{\text{max}}^l\right)=0, \forall l.
\end{equation}
Therefore, the optimal value of $p_l$ becomes
\begin{equation}\label{eqn:optimal pl}
   p_l^{\text{opt}} \!=\! \text{min}\!\left(\!\!\frac{\alpha_2 \sqrt{\keciUEUL\kecitUL} \betlUL\mu_l^{\text{UL}}\Re(\ul^H\gb_l)}{\alpha_1 \! \Ksum \betkDL\abss{\fb}\mu_k^{\text{DL}}\abss{\uk} \!+\! \alpha_2 \kecitUL\sum\limits_{l^{\prime} \in \L} \betlpUL \mu_{l^{\prime}} \abss{\gb_{l}^H \boldsymbol{u}_{l^{\prime}}}\!+\! \bar{\kecitUL}\abss{\gb_l}\jsum\mu_{j^{\prime}} \abss{\boldsymbol{u}_j}},\sqrt{P_{\text{max}}^l}\!\!\right)\!.
\end{equation}

According to the above analysis, the  alternating procedure to solve the sub-problems of $\pr_2$ are summarized in Algorithm 1.

\begin{algorithm}[t]
\small
    \caption{Iterative Algorithm to Solve $\mathcal{P}_2$ Given in (\ref{problem:p2main})}
    \label{Algorithm:WMMSE}
    \textbf{Input}: Maximum powers $P_{\text{max}}^{\text{BS}}$, $P_{\text{max}}^l,\forall l$. Channel coefficients $\hb_k,\forall k$, $\gb_l, \forall l$, $\fb, \forall l,k$. Initial values for $p_l^{(0)}$, $\w_k^{(0)}$ and stopping accuracy $\epsilon_1$.
    \begin{algorithmic}[1]
    \For {$n=1,2,...$},
            \State Update $\uk^{(n)}$ using (\ref{uk}).
            \State Update $\ul^{(n)}$ using (\ref{ul}).
            \State Update $\mu_k^{\text{DL}}$ using (\ref{muk}) while $e_k^{\text{DL}}$ is computed as in (\ref{ek}).
            \State Update $\mu_l^{\text{UL}}$ using (\ref{mul}) while $e_l^{\text{UL}}$ is computed as in (\ref{el}).
            \If{$J(0) \leq P_{\text{max}}^{\text{BS}}$}
                \State Update $\w_k^{(n)}=\w_k(0)$ as in (\ref{eqn: wk optimal with lambda})
                \Else \ Find $\lambda^{\text{opt}}$ using bi-section search, and update $\w_k^{(n)}$ using (\ref{eqn:wOpt}).
                \EndIf
            \State Update $p_l^{(n)}$ using (\ref{eqn:optimal pl}).
      \State \textbf{Until} $\left | \text{SWSR}^{(n)}-\text{SWSR}^{(n-1)}\right | < \epsilon_1$
      \EndFor
   \end{algorithmic}
    \textbf{Output}: The optimal solutions: $\w_k^{\text{opt}}=\w_k^{(n)}, \ \forall k$, $\ul^{\text{opt}}=\ul^{(n)}$ and $\rho_l^{\text{opt}}=(p_l^{(n)})^2, \  \forall l$.
\end{algorithm}
\subsection{Optimizing $\boldsymbol{\hat{\Theta}}$ by Gradient Method}\label{sec:Gradient}
In this section, we solve the main problem $\pr_1$ to optimize the phase shifts of IRSs for a given set of the DL beamformer,  the UL combining vector and the transmit UL power of the users. In the following, we firstly reformulate the modified optimization problem and then solve the problem by using the gradient approach.

\subsubsection{\textcolor{black}{Optimization Problem Transformation}}
Let us define $\bphi=\operatorname{diag}\left(\phi_{1,1},\ldots,\phi_{1,M_1},\ldots,\phi_{R,M_R}\right)$ such that ${\theta}_{r,m}=e^{j\phi_{r,m}}$. Then, we can present $\TH$ as a function of $\boldsymbol{\Phi}$, i.e. $\THPhi$. For the sake of simplicity of indices, we assume that $\phi_n$ is the $n$-th diagonal element of the matrix $\bphi \in \mathbb{H}^{M}$. Therefore, for a given set of variables ${\lbrace\w_k,\ul,p_l \rbrace}$, the optimization problem $\pr_1$ is presented by the following unconstrained version.
\begin{maxi}
	{\boldsymbol{\Phi}}{\mathcal{F}(\THPhi)\label{problem:p3}}
	{}{\mathcal{P}_3: \ \ }
\end{maxi}
where
\begin{equation}\label{eqn:theta transformed objective}
\begin{split}
    \mathcal{F}(\THPhi) &=\!  \alpha_1 \sum\limits_{k\in \K}\betkDL \log_2\left(\!1+\frac{\keciUEDL\kecitDL\B_{k,k}}{\kecitDL\isum\B_{i,i\neq k}\!+\!\bar{\keciUEDL}\kecitDL\B_{k,k}\!+\!\bar{\kecitDL}\isum\abss{\w_i}\Q_k\!+\!\Lsum\C_{l,k}\!+\!\NDL}\!\right) \\
    &+ \alpha_2 \sum\limits_{l\in \L} \betlUL \log_2\left(\!1+\frac{\keciUEUL\kecitUL\Bt_{l,l}}{\kecitUL\jsum\Bt_{j,j\neq l}\!+\!\bar{\keciUEUL}\kecitUL\Bt_{l,l}\!+\!\bar{\kecitUL}\abss{\ul}\jsum\T_j\!+\!\RSI\!+\!\NUL}\!\right).
\end{split}
\end{equation}
In the following, we explain the new appeared variables $\B_{k,i}$, $\C_{l,k}$, $\Bt_{l,j}$, $\Q_k$, and $\T_l$ in (\ref{eqn:theta transformed objective}) based on the quadratic terms in \eqref{eqn: DLSINR} and \eqref{eqn: ULSINR} and reformulate the SWSR in form of some quadratic terms at the DL/UL SINR, separately. For the DL SINR, by using change of variables $\hhat_k^H\THPhi\Hhat\w_i=\vHphi\aki$, where $\vphi=\left[\theta_1,\theta_2,...,\theta_{M_R}\right]^T$, $\aki=\operatorname{diag}(\hhat_k^H)\Hhat\w_i$ and $\cki=\h_k^H\w_i$, the quadratic term of the DL SINR in \eqref{eqn:theta transformed objective} is rewritten as
\begin{equation}\label{eqn:quad term in DL rate}
    \begin{split}
        \B_{k,i}&={\lvert{(\boldsymbol{h}_k^H +\hat{\boldsymbol{h}}_k^H\hat{\boldsymbol{\Theta}}\hat{\boldsymbol{H}})\boldsymbol{w}_i}\rvert}^2 = \vHphi\aki\aki^H\vphi+2\Re(\vHphi\aki\cki^*)+\abss{\cki} \\
        & = \sum\limits_{t=1}^{M} \ \extm \sum\limits_{u=1}^M\atilut \ \exup +2\Re\left(\sum\limits_{u=1}^M \Tilde{\boldsymbol{c}}_u \ \exup\right) + \abss{\cki},
    \end{split}
\end{equation}
where $\atilut$ is $(u,t)$-th entry of matrix $\Tilde{\boldsymbol{a}}_{k,i} = \aki\aki^H \in \mathbb{H}^{M}$ and $\Tilde{\boldsymbol{c}}_u$ is $u$-th element of vector $\Tilde{\boldsymbol{c}} = \aki\cki^* \in \mathbb{C}^{M\times 1}$, wherein index $k$ and $i$ have been eliminated for simplicity of notations.
\textcolor{black}{By using change of variables $\hhat^H\THPhi\Hhat=\vHphi\boldsymbol{m}_k$ where $\boldsymbol{m}_k=\operatorname{diag}(\hhat_k^H)\Hhat$, the quadratic term related to the hardware impairment of the DL SINR in \eqref{eqn:theta transformed objective} is given by
\begin{equation}\label{eqn:quad term related to HI in DL SINR}
    \begin{split}
        \Q_k&={\lvert{(\boldsymbol{h}_k^H +\hat{\boldsymbol{h}}_k^H\hat{\boldsymbol{\Theta}}\hat{\boldsymbol{H}})}\rvert}^2 = \vHphi\boldsymbol{m}_k\boldsymbol{m}_k^H\vphi+2\Re(\vHphi\boldsymbol{m}_k\h_k)+\abss{\h_k} \\
        & = \sum\limits_{t=1}^{M} \ \extm \sum\limits_{u=1}^M\mtilut \ \exup +2\Re\left(\sum\limits_{u=1}^M \ytil_u \ \exup\right) + \abss{\h_k},
    \end{split}
\end{equation}
where $\mtilut$ is $(u,t)$-th entry of matrix $\Tilde{\boldsymbol{m}}_k = \boldsymbol{m}_k\boldsymbol{m}_k^H \in \mathbb{H}^{M}$ and $\ytil_u$ is $u$-th element of vector $\Tilde{\boldsymbol{y}}=\boldsymbol{m}_k\h_k \in \mathbb{C}^{M\times 1}$.}
Moreover, by applying change of variables $\blk=\operatorname{diag}(\hhat_k^H)\ghat_l\sqrt{\rho_l}$ and $\fonelk=\flk\sqrt{\rho_l}$, the interference term given in the denominator of (\ref{eqn:theta transformed objective}) is reformulated by
\begin{equation}\label{eqn:quad interference term in DL rate}
    \begin{split}
       \C_{l,k}&= {\lvert{(\flk +\hat{\boldsymbol{h}}_k\hat{\boldsymbol{\Theta}}\hat{\boldsymbol{g}}_l)}\rvert}^2 \rho_l = \vHphi\blk\blk^H\vphi+2\Re(\vHphi\blk\fonelk^*)+\abss{\fonelk} \\
        & = \sum\limits_{t=1}^M \ \extm \sum\limits_{u=1}^M\btilut \ \exup +2\Re\left(\sum\limits_{u=1}^M \Tilde{\boldsymbol{f}}_u \ \exup\right) + \abss{\fonelk},
    \end{split}
\end{equation}
where $\Tilde{\boldsymbol{b}}_{l,k} = \blk\blk^H \in \mathbb{H}^{M}$ and $\Tilde{\boldsymbol{f}}_u $ is $u$-th element of vector $\Tilde{\boldsymbol{f}} = \blk\fonelk^* \in \mathbb{C}^{M\times 1}$.

Similarly for the UL SINR given in (\ref{eqn:theta transformed objective}), by using change of variables $\ul^H\Hhat^H\THPhi\ghat_j\sqrt{\rho_j}=\vHphi\zlj$, where $\zlj =\operatorname{diag}(\ul^H\Hhat^H)\ghat_j\sqrt{\rho_j}$ and $\dlj=\ul^H\g_j\sqrt{\rho_j}$, we have
\begin{equation}\label{eqn:quad term in UL rate}
    \begin{split}
       \Bt_{l,j}&= \abss{\ul^H\g_j+\ul^H\Hhat^H\THPhi\ghat_j}\rho_j = \vHphi\zlj\zlj^H\vphi+2\Re(\vHphi\zlj\dlj^*)+\abss{\dlj} \\
        & = \sum\limits_{t=1}^M \ \extm \sum\limits_{u=1}^M\ztilut \ \exup +2\Re\left(\sum\limits_{u=1}^M \Tilde{\boldsymbol{d}}_u \ \exup\right) + \abss{\dlj},
    \end{split}
\end{equation}
where $\ztilut$ is $(u,t)$-th entry of matrix $\Tilde{\boldsymbol{z}}_{l,j}=\zlj\zlj^H \in \mathbb{H}^{M}$ and $\dtil$ is $u$-th element of vector $\zlj\dlj^* \in \mathbb{C}^{M\times 1}$. \textcolor{black}{Moreover, for quadratic term related to hardware impairment of the UL SINR in \eqref{eqn:theta transformed objective}, by applying change of variables $\Hhat^H\THPhi\ghat_j=\stil_j\vphi$ where $\stil_j=\Hhat^H\operatorname{diag}(\ghat_j)\sqrt{\rho_j}$, we have
\begin{equation}\label{eqn:quad term related to HI in UL SINR}
    \begin{split}
       \T_j&= \abss{\g_j+\Hhat^H\THPhi\ghat_j}\rho_j = \vHphi\stil_j^H\stil_j\vphi+2\Re(\vHphi\stil_j^H\g_j\sqrt{\rho_j})+\abss{\g_j}\rho_j \\
        & = \sum\limits_{t=1}^M \ \extmT \sum\limits_{u=1}^M\xtilut \ \exupT +2\Re\left(\sum\limits_{u=1}^M \Tilde{\boldsymbol{e}}_u \ \exupT\right) + \abss{\g_j}\rho_j,
    \end{split}
\end{equation}
where $\xtilut$ is $(u,t)$-th entry of matrix $\Tilde{\boldsymbol{x}}_j=\stil_j^H\stil_j \in \mathbb{H}^M$ and $\etil_u$ is $u$-th element of vector $\stil^H\g_j \sqrt{\rho_j} \in \mathbb{C}^{M\times 1}$.}

\subsubsection{Gradient-based  Approach for $\mathcal{P}_3$}
As shown in \cite{huang2019reconfigurable}, the gradient-based search approach can be used to obtain appropriate phase shifts of the IRS. Here, we also apply this approach for multiple IRSs to solve the optimization problem $\pr_3$. \textcolor{black}{It is worth noting that this approach is not guaranteed to converge to a globally optimal value; however, locally optimal phase shift matrices are obtained.}
To this end, let us assume that $\boldsymbol{\Phi}^{(s)}$ denotes the phase vector at $s$-th iteration. Thus, the next iteration point is given by
\begin{equation}\label{gradient Ascent}
    \boldsymbol{\Phi}^{(s+1)} = \boldsymbol{\Phi}^{(s)} + \eta \nabla_{\boldsymbol{\Phi}}\mathcal{F}(\THPhiGDA),
\end{equation}
where 
\begin{equation}\label{gradient}
    \nabla_{\boldsymbol{\Phi}}\mathcal{F}(\THPhi) = \left[\frac{\partial \mathcal{F}(\THPhi)}{\partial\phi_1},\ldots,\frac{\partial \mathcal{F}(\THPhi)}{\partial\phi_{M_R}}\right]^T,
\end{equation}
is the gradient of the objective function (\ref{eqn:theta transformed objective}) and $\eta$ is the step size which can be found efficiently at each step by using backtracking line search based on the Armijo–Goldstein condition \cite{boyd2004convex}.  Subsequently, the elements of the gradient vector (\ref{gradient}) are evaluated as
\begin{equation}\label{derivative}
    \frac{\partial \mathcal{F}(\THPhi)}{\partial\phi_n} = \alpha_1\Ksum\betkDL \frac{\partial R_k^{\text{DL}}}{\partial \phi_n} +\alpha_2\Lsum\betlUL\frac{\partial R_l^{\text{UL}}}{\partial \phi_n}, 
\end{equation}
where $\frac{\partial R_k^{\text{DL}}}{\partial \phi_n}$ and $\frac{\partial R_l^{\text{UL}}}{\partial \phi_n}$ are given by \begin{equation*}
\hspace*{-15Cm}
\frac{\partial R_k^{\text{DL}}}{\partial \phi_n} = 
\end{equation*}
\begin{equation}\label{Der DL}
 =\frac{F_1\B_{k,k}^{\prime}\!\left(\!\kecitDL\isum\B_{k,i\neq k}\!+F_3\Q_k\!+\Lsum\C_{l,k}\!+\NDL\!\right)-F_1\B_{k,k}\!\left(\!\kecitDL\isum\B_{k,i\neq k}^{\prime}\!+F_3\Q_k^{\prime}\!+\Lsum\C_{l,k}^{\prime}\!\right)}{\left(\!\kecitDL\isum\B_{k,i\neq k}\!+F_2\B_{k,k}\!+F_3\Q_k\!+\Lsum\C_{l,k}\!+\NDL\!\right)\!\left(\!\kecitDL\isum\B_{k,i}\!+F_3\Q_k\!+\Lsum\C_{l,k}\!+\NDL\right)},
\end{equation}

\begin{equation*}
    \hspace*{-15Cm}    \frac{\partial R_l^{\text{UL}}}{\partial \phi_n} =
\end{equation*}
\begin{equation}\label{Der UL}
=  \frac{E_1\Bt_{l,j}^{\prime}\!\left(\!\kecitUL\jsum\Bt_{l,j\neq l}\!+E_3(\ul)\jsum\T_j\!+E_4(\ul)\!\right)-E_1\Bt_{l,l}\!\left(\!\kecitUL\jsum\Bt_{l,j\neq l}^{\prime}\!+E_3(\ul)\jsum{\T_j^{\prime}}\!\right)}{\left(\!\kecitUL\jsum\Bt_{l,j\neq l}\!+E_3(\ul)\jsum\T_j\!+E_4(\ul)\!\right)\!\left(\!\kecitUL\jsum\Bt_{l,j}\!+E_2\Bt_{l,l}\!+E_3(\ul)\jsum\T_j\!+E_4(\ul)\!\right)},
\end{equation}
where $F_1=\keciUEDL\kecitDL$, $F_2=\bar{\keciUEDL}\kecitDL$, $F_3=\bar{\kecitDL}\isum\abss{\w_i}$, $E_1=\keciUEUL\kecitUL$, $E_2=\bar{\keciUEUL}\kecitUL$, $E_3(\ul)=\bar{\kecitUL}\abss{\ul}$, and $E_4(\ul)=\RSI+\NUL \abss{\ul}$.
Thus, based on the fact that $\Tilde{\boldsymbol{a}} = \Tilde{\boldsymbol{a}}^H$, $\Tilde{\boldsymbol{b}} = \Tilde{\boldsymbol{b}}^H$,
$\Tilde{\boldsymbol{m}} = \Tilde{\boldsymbol{m}}^H$,
$\Tilde{\boldsymbol{z}} = \Tilde{\boldsymbol{z}}^H$, and $\Tilde{\boldsymbol{x}} = \Tilde{\boldsymbol{x}}^H$,
derivative of each of the defined terms given in (\ref{eqn:quad term in DL rate})--(\ref{eqn:quad term related to HI in UL SINR}) are presented as  
    \begin{equation}\label{DerB}
        \B_{k,i}^{\prime}=\frac{\partial\B_{k,i}}{\partial{\phi_n}} = 2\Re\left(j\exnp\left(\Tilde{\boldsymbol{c}_n}+\sum\limits_{t\neq n}^{M}\atilnt\extm\right)\right),
    \end{equation}
    \begin{equation}\label{DerC}
        \C_{l,k}^{\prime}=\frac{\partial\C_{l,k}}{\partial{\phi_n}} = 2\Re\left(j\exnp\left(\Tilde{\boldsymbol{f}_n}+\sum\limits_{t\neq n}^{M}\btilnt\extm\right)\right),
    \end{equation}
    \begin{equation}\label{DerQ}
    \Q_{k}^{\prime}=\frac{\partial\Q_{k}}{\partial{\phi_n}} = 2\Re\left(j\exnp\left(\Tilde{\boldsymbol{y}_n}+\sum\limits_{t\neq n}^{M}\mtilnt\extm\right)\right),
    \end{equation}
    \begin{equation}\label{DerBt}
        \Bt_{l,j}^{\prime}=\frac{\partial\Bt_{l,j}}{\partial{\phi_n}} = 2\Re\left(j\exnp\left(\Tilde{\boldsymbol{d}_n}+\sum\limits_{t\neq n}^{M}\ztilnt\extm\right)\right),
    \end{equation}
        \begin{equation}\label{DerT}
        \T_{j}^{\prime}=\frac{\partial\T_{j}}{\partial{\phi_n}} = 2\Re\left(-j\exnpT\left(\Tilde{\boldsymbol{e}_n}+\sum\limits_{t\neq n}^{M}\xtilnt\extmT\right)\right).
    \end{equation}
\begin{algorithm}[t]
\small
    \caption{Iterative Algorithm for Solving $\pr_1$}
    \label{Algorithm:overall}
    \textbf{Input}: Initial value for $\TH^{(0)}$, maximum powers $P_{\text{max}}^{\text{BS}}$, $P_{\text{max}}^l,\forall l$, channel coefficients $\hb_k,\forall k$, $\gb_l, \forall l$, $\fb, \forall l,k$. Initial values for $p_l^{(0)}$, $\w_k^{(0)}$, tolerances $\lbrace{\epsilon_1,\epsilon_2,\epsilon_3}\rbrace$.
    \begin{algorithmic}[1]
         \For {$i=1,2,...$}
         \State For given $\w_k^{(i)},\forall k$, $\ul^{(i)}$ and $\rho_l^{(i)}, \forall l$, update $\TH^{(i)}$
         \For {$s=0,1,...$}
         \State Update ascent direction $\boldsymbol{d}^{(s)}$using (\ref{gradient}).
         \State Update $\boldsymbol{\Phi}^{(s+1)}$ using (\ref{gradient Ascent})
         \State \textbf{Until} $\normm{\nabla_{\boldsymbol{\Phi}}\mathcal{F}(\THPhiGDA)} < \epsilon_2$; Obtain $\TH^{(i+1)}=\TH^{(s+1)}$
         \EndFor
         \State For given $\TH^{(i+1)}$ update $\w_k^{(i+1)},\forall k$, $\ul^{(i+1)}$ and $\rho_l^{(i+1)}, \forall l$ by solving problem $\pr_2$ using \textbf{Algorithm}1.
         \State \textbf{Until} $\left  | \text{SWSR}^{(i+1)}-\text{SWSR}^{(i)}\right |<\epsilon_2$.
         \EndFor  
    \end{algorithmic}
    \textbf{Output}: The optimal solutions: $\w_k^{\text{opt}}=\w_k^{(i)},\forall k$, $\ul^{\text{opt}}=\ul^{(i)}$,  $\rho_l^{\text{opt}}=\rho_l^{(i)}, \forall l$ and $\TH^{\text{opt}}=\TH^{(i+1)}$.
\end{algorithm}
\subsection{Complexity Analysis}\label{sec: complexity}
In the preceding sections, we investigated the maximization of the SWSR, i.e. the optimization problem $\pr_1$, by transforming the original  optimization problem into several sub-problems. The whole iterative procedure for solving problem $\pr_1$ is  summarized in Algorithm 2. In \textbf{Algorithm 1}, the complexity of computing $\ul$ in step 3 is $\mathcal{O}(N_t^3)$, complexity of computing $\w_k$ in step 7 which performs matrix inversion is $\mathcal{O}(N_t^3)$. Moreover, \textcolor{black}{by assuming that the step 8 is used to  find the $\w_k^{\text{opt}}$, the complexity of the eigenvalue decomposition of matrix $A$ which is used in step 8 is $\mathcal{O}(N_t^3)$; also, the bi-section search to find $\lambda^{\text{opt}}$ in step 8 adds $\mathcal{O}(\log(\frac{\lambda_{\text{max}}-\lambda_{\text{min}}}{\varepsilon}))$ complexity where $\varepsilon$ is the error tolerance.} The number of iterations in this algorithm is denoted by $I_{\mathrm{mse}}$.
As for \textbf{Algorithm 2}, we can see that optimizing the phase shifts of the IRSs relies on the number of $I_{\mathrm{ga}}$, i.e. the gradient ascent iteration. Thus, the complexity of problem $\pr_3$ is $I_{\mathrm{ga}}\mathcal{O}(M^2)$. Therefore, the total complexity of \textbf{Algorithm 2} is $\mathcal{O}\left(I_{\mathrm{tol}}\left(I_{\mathrm{mse}}\left(3 N_t^3+\log(\frac{\lambda_{\text{max}}-\lambda_{\text{min}}}{\varepsilon})\right)+I_{\mathrm{ga}}M^2\right)\right)$ where $I_{\mathrm{tol}}$ represents the number of iteration in \textbf{Algorithm 2}. 
The iterative WMMSE approach developed in \textbf{Algorithm 1} is based on block coordinate descent (BCD) method and its convergence is guaranteed as discussed in \cite{shi2011iteratively}.
Meanwhile, the provided numerical results approves that the \textbf{Algorithm 2} converges in a few iterations.
\section{Numerical Results} \label{sec:simulation}
In this section, numerical results are presented to highlight performance of the proposed system for various examples. It is assumed that the BS is equipped with a uniform linear array with $N_t =4$ antennas and is located at $(0,0)$, also two IRSs are deployed; one is located at  $(100 \text{m},0 \text{m})$ and the second one is placed at $(-100 \text{m},0 \text{m})$. Moreover, the number of uplink  and downlink users are $L=3$ and  $K=2$, respectively. The large scale path loss is modeled by $\text{PL}=-35.6-10\alpha\log_{10}(d)$ dB, wherein $d$ is the relative distance between transmitter--receiver pair and the path loss exponents are $\alpha_{\text{BI}}=2.1$ for the BS-IRSs links,   $\alpha_{\text{IU}}=2.2$ for the IRS--user, $\alpha_{\text{BU}}=4$ for the BS--users, and  $\alpha_{\text{UU}}=3.1$ for user--user. Since the IRSs are usually deployed in practice near the BS or near the users \cite{Larsson-May2020-Weighted-SR}, the corresponding channels are modeled as line-of-sight (LOS) ones. Therefore, the small-scale channels $\Tilde{\boldsymbol{H}}_r$ where $r\in \lbrace{1,2}\rbrace$ are modeled by Rician fading as follows
\begin{equation}\label{Rician H}
    \Tilde{\boldsymbol{H}}_r=\sqrt{\frac{\kappa}{1+\kappa}}\boldsymbol{a}_{M}(\vartheta^{\text{AoA}})\boldsymbol{a}_{N_t}^H(\vartheta^{\text{AoD}}) + \sqrt{\frac{1}{1+\kappa}} \boldsymbol{H}^{\text{NLOS}},
\end{equation}
and for small-scale channels $\Tilde{\h} \in \lbrace{\Tilde{\hsss},\Tilde{\gs}}\rbrace$, we have
\begin{equation}\label{Rician h,g}
    \Tilde{\h}=\sqrt{\frac{\kappa}{1+\kappa}}\boldsymbol{a}_{M}(\vartheta^{\text{AoA}}) + \sqrt{\frac{1}{1+\kappa}} \h^{\text{NLOS}},
\end{equation}
where $\kappa=4$ denotes the Rician factor and $\boldsymbol{a}_{M}(\vartheta^{\text{AoA}})\boldsymbol{a}_{N_t}^H(\vartheta^{\text{AoD}})$ represents the LOS component. Variables $\vartheta^{\text{AoA}}$ and $\vartheta^{\text{AoD}}$ denote the angle-of-arrival (AoA) and angle-of-departure (AoD) of IRSs which are  uniformly distributed over $[0,2\pi)$, respectively. The term  {$\boldsymbol{a}_{n} \in \mathbb{C}^{n \times 1}$ denotes  the steering vector} and is defined as
\begin{equation}\label{steering vector}
\boldsymbol{a}_{n}=\left[1, e^{j \frac{2 \pi D}{\lambda_r} \sin \vartheta}, \ldots, e^{j \frac{2 \pi D}{\lambda_r}(n-1) \sin \vartheta}\right]^{T},
\end{equation}
where $D$ is the antenna element separation, $\lambda_r$ is the carrier wavelength and $\frac{D}{\lambda_r}=1/2$ is used. The NLOS components, i.e. $\boldsymbol{H}^{\text{NLOS}}$, $\boldsymbol{h}^{\text{NLOS}}$ and channels between users are modeled by  zero-mean and unit variance Rayleigh distribution RVs. \textcolor{black}{It is also assumed that the hardware are perfect, unless is further specified with providing the hardware quality factors.} Also, the weighting parameters $\alpha_1$ and $\alpha_2$  are assumed to be one unless their values are provided. Other simulation parameters are listed in Table 1.

\begin{table}[]
\centering
\caption{Numerical Results Parameters}
\label{tab:my-table}
\scalebox{1}{
\begin{tabular}{|l|l|}
\hline
\textbf{Parameters} & \textbf{Values}\\ \hline
Maximum transmission power at BS, $P_{\text{max}}^{\text{BS}}$ & 35 [dBm] \\ \hline
 Maximum transmission power at UL users, $P_{\text{max}}^l,\forall l$& 11 [dBm] \\ \hline
Noise power at DL users & -100 [dBm] \\ \hline
  Algorithms convergence parameters, $\lbrace{\epsilon_1,\epsilon_2,\epsilon_3}\rbrace$ & $\lbrace{10^{-3},10^{-4},10^{-3}}\rbrace$ \\ \hline
Weights of DL and UL users $\betkDL,\betlUL, \forall k,l$& 1 \\ \hline
Residual self interference channel variance. $\hat{\sigma}^2$ & -95 [dBm] \\ \hline
 Noise power at BS & -110 [dBm] \\ \hline 
 Rician factor for reflecting links & 6 [dB] \\ \hline
\end{tabular}
}
\end{table}

\subsection{Convergence}
\begin{figure}[!t]
   \centering
   \begin{minipage}[b]{0.45\textwidth}
   \psfrag{iteration}[][][0.76]{Number of Iterations}
  \psfrag{ssr}[][][0.76]{SWSR [bpcu]}
  \psfrag{Nt=4 Mt=10 K=1 L= 3 AAAAAAAAAA}[][][0.6]{\ \ $N_t=4$, $M=10$, $K=1$, $L= 3$}
  \psfrag{Nt=2 Mt=4 K=2 L=3 AAAAAAAAAAAA}[][][0.6]{\!$N_t=2$, $M=4$, $K=2$, $L=3$}
  \psfrag{Nt=4 Mt=10 K=1 L=1AAAAAAAAAAAA}[][][0.6]{$N_t=4$, $M=10$, $K=1$, $L=1$}
  \psfrag{Nt=2 Mt=4 K=1 L=1AAAAAAAAAAAAA}[][][0.6]{\!\!\!\!$N_t=2$, $M=4$, $K=1$, $L=1$}
  \hspace*{-0.72 cm}
\includegraphics[scale=0.63]{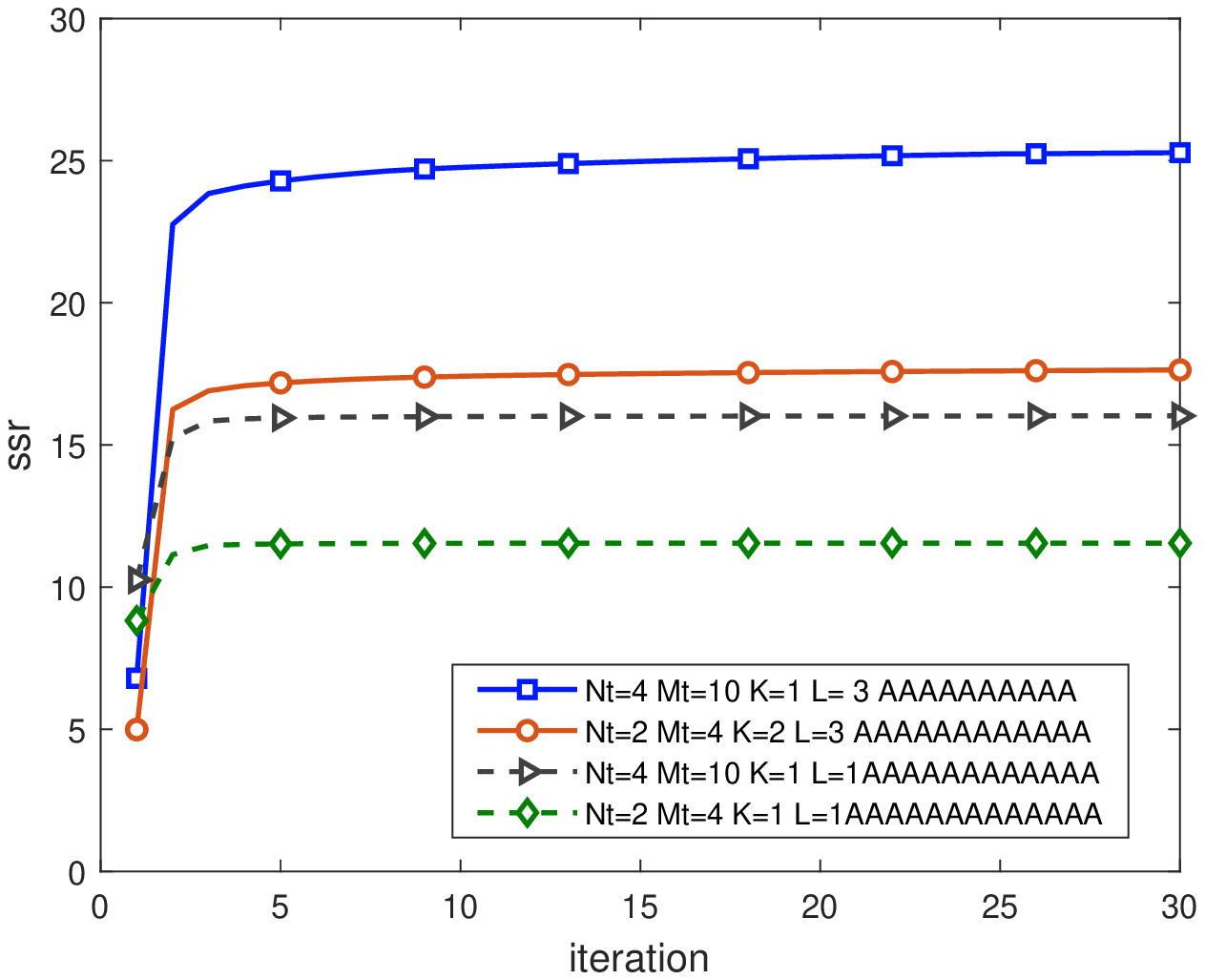}
\caption{Convergence of the proposed algorithm.}
\label{fig:convergence}
    \end{minipage}
    \hfill
    \begin{minipage}[b]{0.45\textwidth}
       \centering
   \psfrag{Number of IRS element(M1=M2)}[][][0.76]{Number of elements of each IRS $(M_1=M_2)$}
   \psfrag{System Sum Rate [bits/s/Hz]}[][][0.76]{SWSR [bpcu]}
   \psfrag{Scheme 1 AA}[][][0.6]{\!\!\!Scheme 1}
   \psfrag{Scheme 2 AA}[][][0.6]{\!\!\!Scheme 2}
   \psfrag{Scheme 3 AA}[][][0.6]{\!\!\!Scheme 3}
   \psfrag{Scheme 1 HD AA}[][][0.6]{Scheme 1 HD}
   \psfrag{Scheme 2 HD AA}[][][0.6]{Scheme 2 HD}
   \psfrag{Scheme 1 HI AA}[][][0.6]{Scheme 1 HI}
   \psfrag{Scheme 2 HI AA}[][][0.6]{Scheme 2 HI}
   \hspace*{-1.2 cm}
   \includegraphics[scale=0.63]{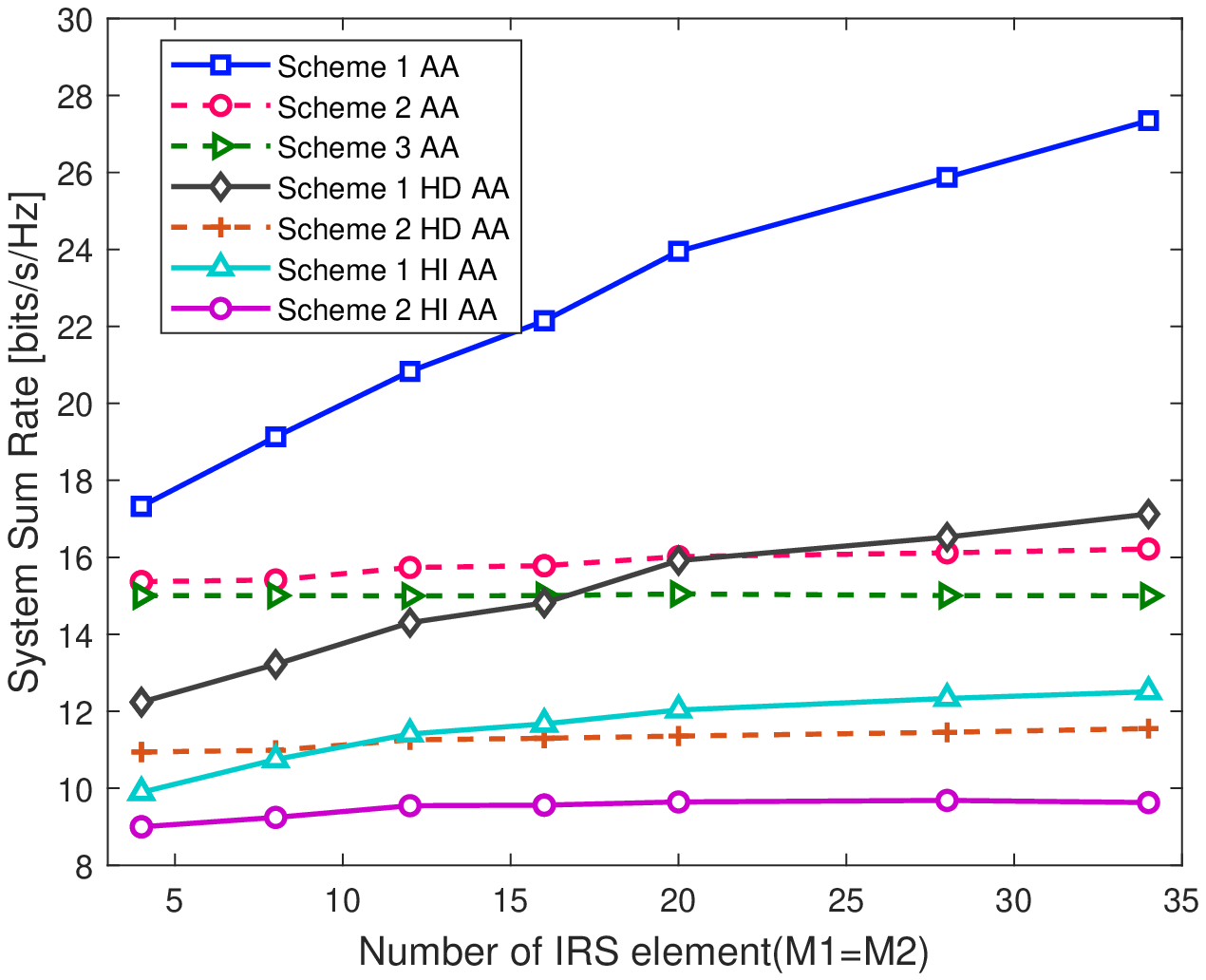}
\caption{System weighted sum-rate versus number of IRSs elements for different schemes.}
\label{fig:SSRvsM}
    \end{minipage}
\end{figure}
Convergence of the proposed  algorithm, i.e, \textbf{Algorithm 2}, is discussed in Fig. \ref{fig:convergence}; the SWSR versus the number of iterations is depicted for various sets of number of IRSs elements, BS antennas, and users at the UL and the DL. 
It is also assumed that the UL users are randomly and uniformly located in a circle centered at $(-100,5)$ with radius of 10 m and, the downlink users are uniformly located in a circle centered at $(100,5)$ with radius of 10 m. As it is shown,  the proposed algorithm converges rapidly for all the sets of parameters. For instance, the algorithm converges in $10$ iterations on average for the case of $\{N_t=4,M=20,K=2,L=3\}$. Also, by increasing the size of IRSs, the BS antennas, and the number of users, it takes more iterations to converge \textcolor{black}{since more optimization variables are involved enlarging the search space for the solution. Besides, as we discussed in Section \ref{sec: complexity}, by increasing the size of IRSs or BS antennas, in each iteration, the proposed algorithm has higher complexity.} 

\subsection{Impact of Number of Reflecting Elements at IRS}
Fig. \ref{fig:SSRvsM} illustrates the SWSR versus the size of IRSs for $M_1=M_2=M$. For comparison, we consider three schemes; scheme 1 denotes the general proposed optimization algorithm, scheme 2 indicates the case wherein the IRSs phase shifts are fixed while the beamformer, UL combining vector  and UL users' powers are optimized by \textbf{Algorithm 1}, and for scheme 3, no IRSs are adopted, i.e. \textit{without} utilizing any IRS ($\vphi=\Vec{\boldsymbol{0}}$) and \textbf{Algorithm 1} is also used for the optimization. Moreover, the HD version of the schemes 1 and 2 are analyzed for benchmarking, i.e, the DL and UL transmissions are performed in two equal time slots and thus there is no self-interference signals as well. It is shown that by increasing the number of elements of IRSs, the proposed algorithm, either in FD mode or HD mode, significantly outperforms the other schemes, since there are more degrees of freedom for customizing the channels between the BS and the IRSs and the channels between the IRSs and the users. Moreover, it can be observed that by exploiting the BS in FD mode, the gain of increasing $M$ is more beneficial than that of the HD mode. Also, the impact of increasing $M$ is depicted for the proposed algorithm when $\kecitDL=\kecitUL=\keciUEDL=\keciUEUL=0.92$. It is shown that by increasing the number of $M$, the proposed algorithm improves at a lower rate rather than using the ideal hardware. However, it outperforms scheme 2 when hardware impairments are assumed. \textcolor{black}{It can be seen that despite the increasing of IRSs elements, the SWSR is limited by a specific value because of the hardware impairment, whereas this limitation does not exist when devices are perfect.} 
\subsection{Impact of Maximum Transmit Power at DL and UL}
    \psfrag{Scheme 1 AAAf}[][][0.6]{\!\!\!Scheme 1}
    \psfrag{Scheme 2 AAAf}[][][0.6]{\!\!\!Scheme 2}
    \psfrag{Scheme 3 AAAf}[][][0.6]{\!\!\!Scheme 3}
    \psfrag{Scheme 1 HD AAAf}[][][0.6]{Scheme 1 HD}
    \psfrag{Scheme 2 HD AAAf}[][][0.6]{Scheme 2 HD}
    \psfrag{SSR}[][][0.76]{SWSR [bpcu]}
    \psfrag{powerDL}[][][0.76]{$P_{\text{max}}^{\text{BS}}$}
    \psfrag{Scheme 1 DL AAAq}[][][0.6]{Scheme 1 DL}
    \psfrag{Scheme 2 DL AAAq}[][][0.6]{Scheme 2 DL}
    \psfrag{Scheme 3 DL AAAq}[][][0.6]{Scheme 3 DL}
    \psfrag{Scheme 1 UL AAAq}[][][0.6]{Scheme 1 UL}
    \psfrag{Scheme 2 UL AAAq}[][][0.6]{Scheme 2 UL}
    \psfrag{Scheme 3 UL AAAq}[][][0.6]{Scheme 3 UL}
    \psfrag{sumrate}[][][0.76]{Sum-Rate [bpcu]}
\begin{figure}[!t]
   \centering
   \hspace*{-0.75 cm}
    \subfloat[]{\label{fig:SSRvsDLpower}{\includegraphics[scale=0.63]{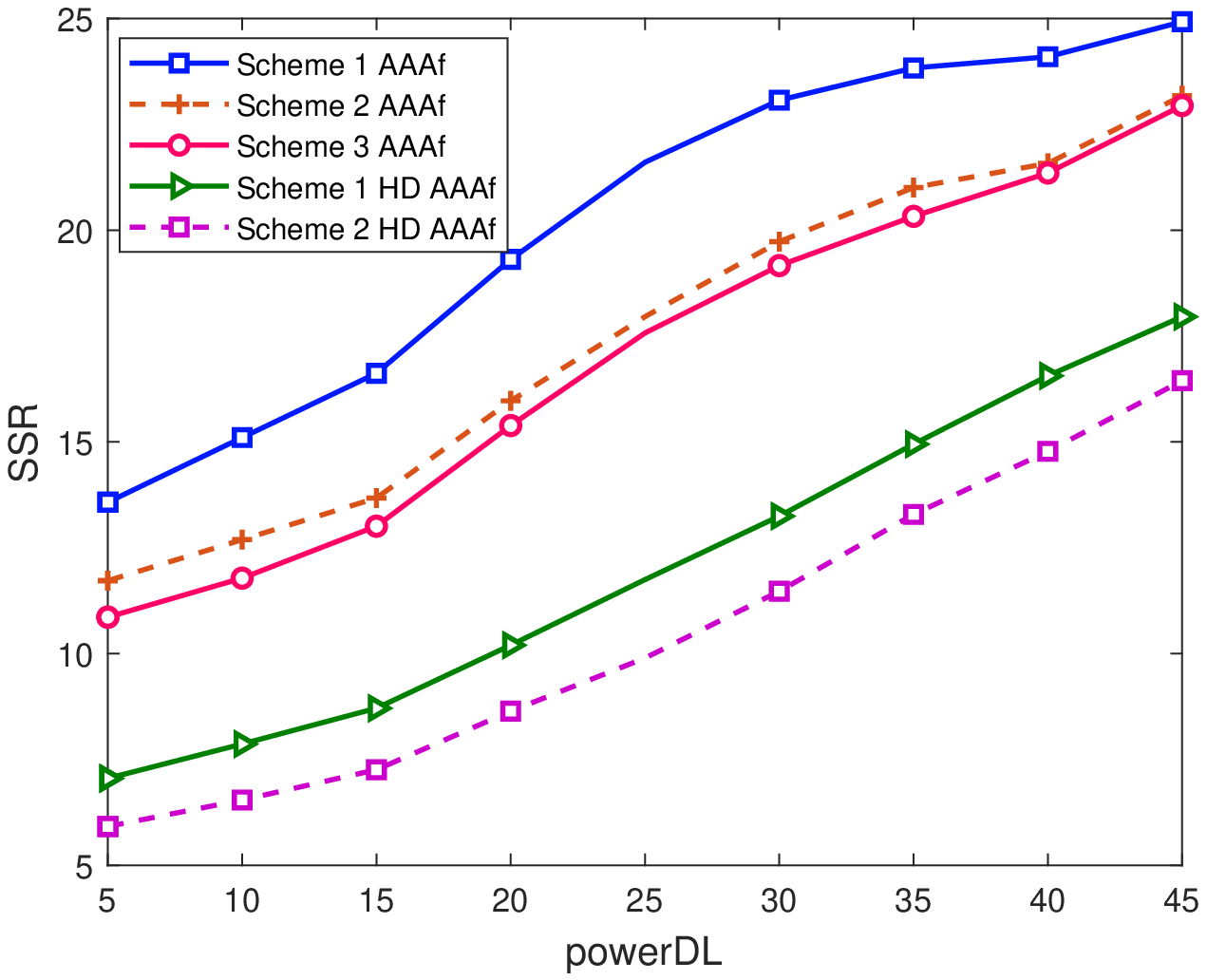} }}
    \hspace*{-0.9 cm}
    \subfloat[]{\label{fig:SumRateVsDLpower}{\includegraphics[scale=0.63]{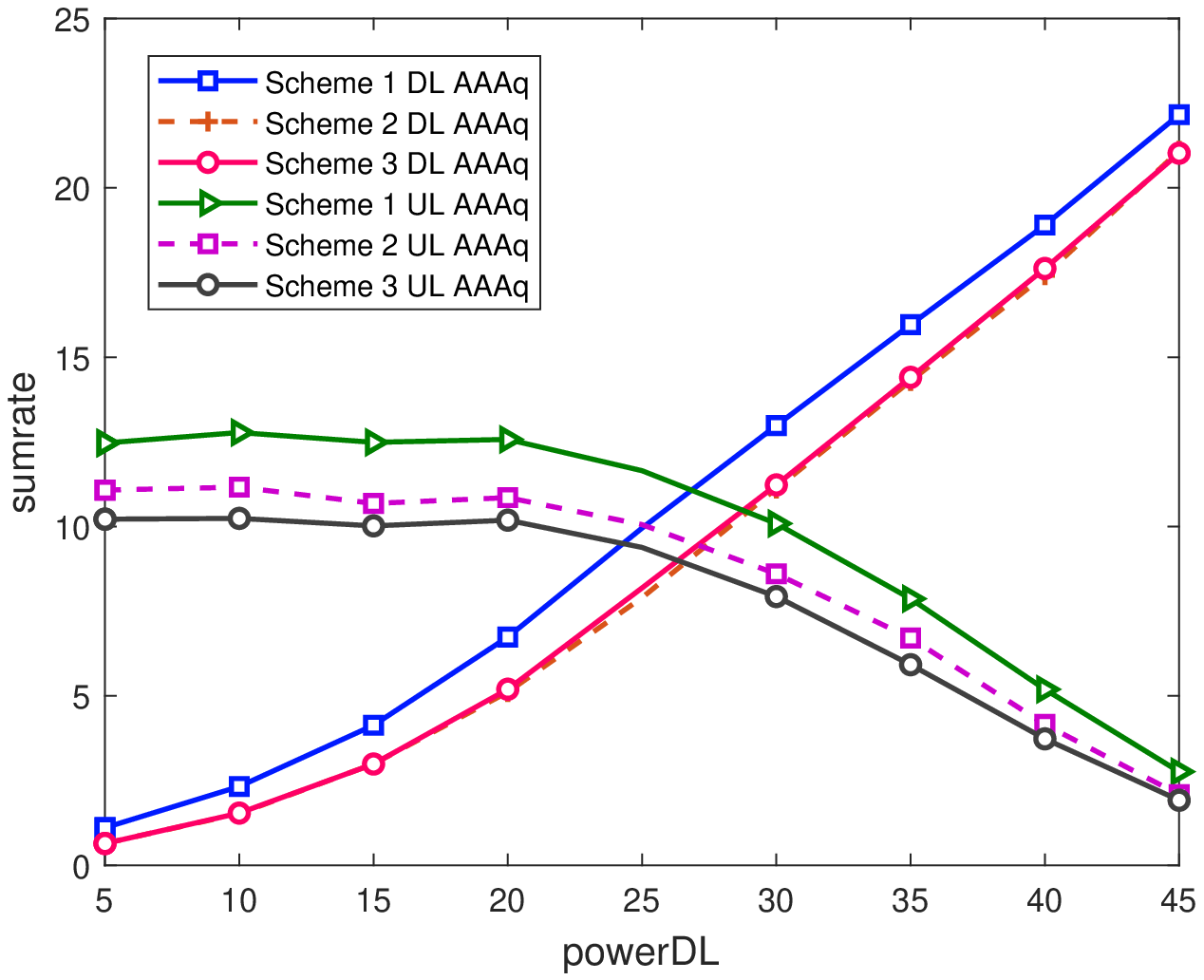} }}    

    \label{fig:DLpower}
    \caption{Performance of different schemes for various maximum transmit power at BS, i.e. $P_{\text{max}}^{\text{BS}}$ for $K=2$, $L=3$, $N_t=4$, $M=14$, and $P_{\text{max}}^l=11$ {[dBm]}. (a) SWSR versus $P_{\text{max}}^{\text{BS}}$. (b) Sum-Rate versus $P_{\text{max}}^{\text{BS}}$.}
\end{figure}

The SWSR as a function of the maximum transmit power of the BS, i.e. $P_{\text{max}}^{\text{BS}}$, is depicted in  Fig. \ref{fig:SSRvsDLpower} and  performance of the three benchmarking schemes with FD and HD scenarios are compared for $K=2$, $L=3$, $N_t=4$, $M=14$, and $P_{\text{max}}^l=11$ \ [dBm]. Although the self-interference power degrades the uplink performance, it is shown that the SWSR improves monotonically by increasing the maximum transmit power of the BS. \textcolor{black}{As a result, to maximize the SWSR for high values of the maximum power of the BS, it seems that the UL users do not participate in the system performance.} Also, the proposed algorithm, i.e. scheme 1 which optimizes the UL-DL parameters along with the phase shifts of the IRSs,  outperforms the others. Moreover, utilizing a FD enlarges the achievable rates about twice. Fig. \ref{fig:SumRateVsDLpower} also individually investigates the effect of $P_{\text{max}}^{\text{BS}}$ on the sum rates of UL and DL data transmissions. In contrast to the sum-rate of downlink, as expected, the sum-rate of the uplink degraded by increasing the maximum transmit power of the BS due to increasing the power of the self-interference.
    \psfrag{Scheme 1 AAAf}[][][0.6]{\!\!\!Scheme 1}
    \psfrag{Scheme 2 AAAf}[][][0.6]{\!\!\!Scheme 2}
    \psfrag{Scheme 3 AAAf}[][][0.6]{\!\!\!Scheme 3}
    \psfrag{Scheme 1 HD AAAf}[][][0.6]{Scheme 1 HD}
    \psfrag{Scheme 2 HD AAAf}[][][0.6]{Scheme 2 HD}
    \psfrag{SSR}[][][0.76]{SWSR [bpcu]}
    \psfrag{Scheme 1 DL AAAe}[][][0.6]{Scheme 1 DL}
    \psfrag{Scheme 2 DL AAAe}[][][0.6]{Scheme 2 DL}
    \psfrag{Scheme 3 DL AAAe}[][][0.6]{Scheme 3 DL}
    \psfrag{Scheme 1 UL AAAe}[][][0.6]{Scheme 1 UL}
    \psfrag{Scheme 2 UL AAAe}[][][0.6]{Scheme 2 UL}
    \psfrag{Scheme 3 UL AAAe}[][][0.6]{Scheme 3 UL}
    \psfrag{sumrate}[][][0.76]{Sum-Rate [bpcu]}
    \psfrag{powerUL}[][][0.76]{$P_{\text{max}}^{\text{UL}}$}
\begin{figure}[!t]
   \centering
   \hspace*{-0.75 cm}
    \subfloat[]{\label{fig:SSRvsULpower}{\includegraphics[scale=0.63]{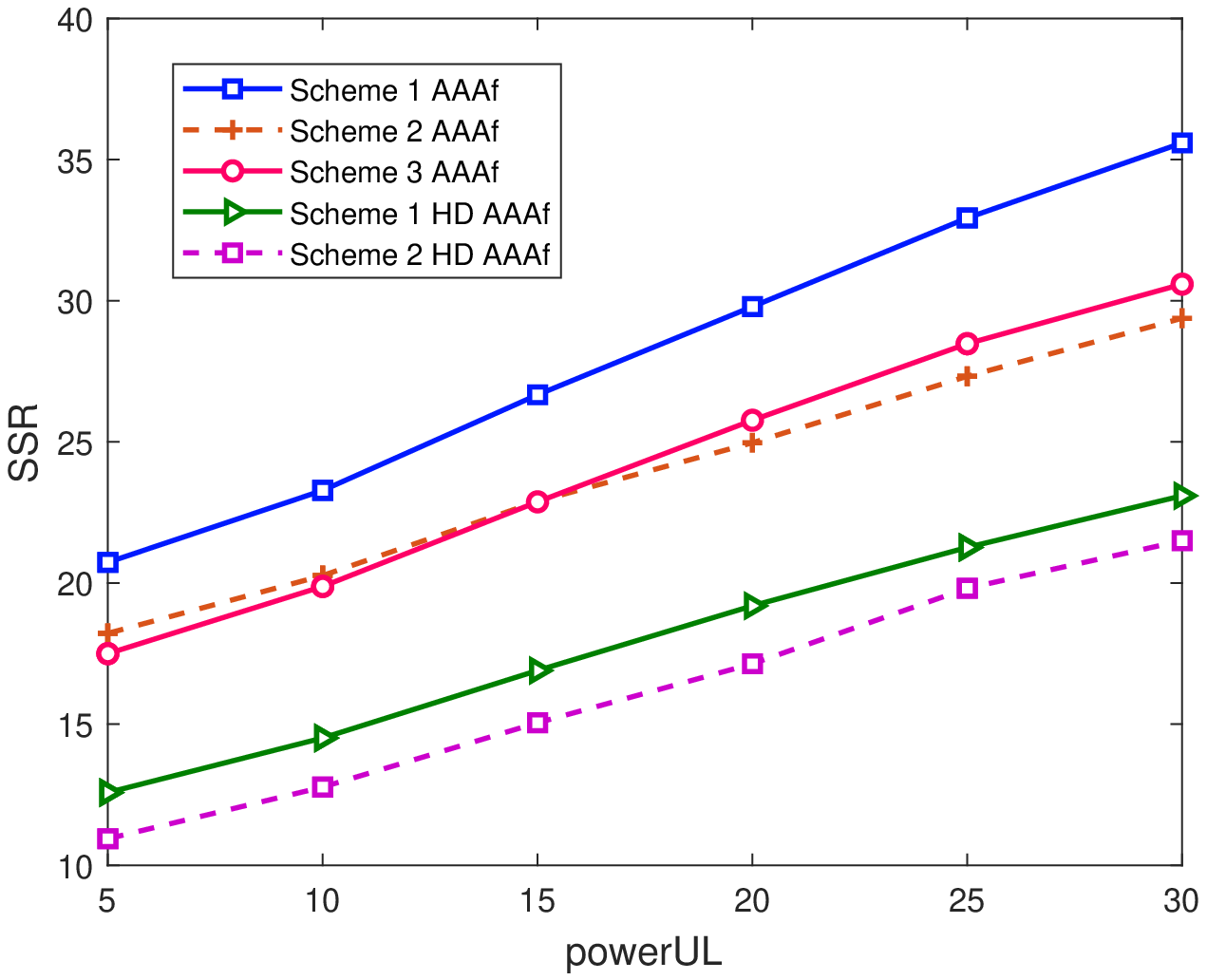} }}
    \hspace*{-0.9 cm}
    \subfloat[]{\label{fig:SumRateVsULpower}{\includegraphics[scale=0.63]{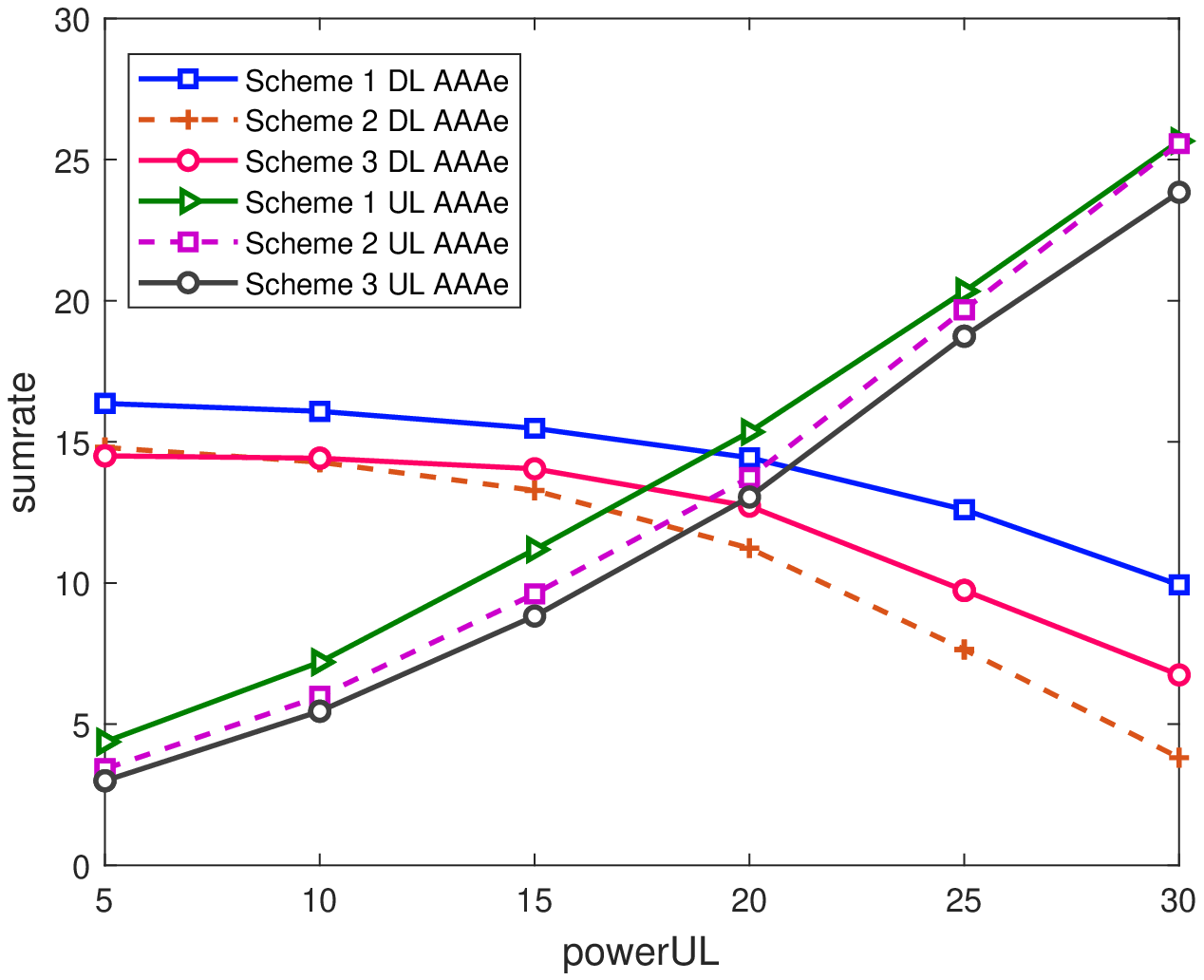} }} 
    \caption{Performance of different schemes for various maximum transmit power at UL users, i.e. $P_{\text{max}}^{\text{UL}}$ for $K=2$, $L=3$, $N_t=4$, $M=14$, and $P_{\text{max}}^{\text{BS}}=35$ {[dBm]}. (a) SWSR versus $P_{\text{max}}^{\text{UL}}$. (b) Sum-Rate versus $P_{\text{max}}^{\text{UL}}$.}
    \label{fig:uplinkPower}    
\end{figure}

Similarly, the effect of the maximum transmit power of UL users on the SWSR and sum rates of UL and DL are investigated in Fig. \ref{fig:uplinkPower}. For the three schemes with FD and HD scenarios, the achievable rates versus $P_{\text{max}}^l=P_{\text{max}}^{\text{UL}}$ is depicted for the case of $K=2$, $L=3$, $N_t=4$, $M=14$, and $P_{\text{max}}^{\text{BS}}=35$ [dBm]. From Fig. \ref{fig:SSRvsULpower}, it can be seen that the SWSR increases by increasing $P_{\text{max}}^{\text{UL}}$ and scheme 1 outperforms the others, and its increasing rate is more. \textcolor{black}{Interestingly, by increasing the maximum transmit power at the UL users, scheme 3, i.e. no-IRS, outperforms scheme 2 with fixed IRS since the non-optimized IRSs in scheme 2 can reflect more power of the signals transmitted by UL users to undesired points. Consequently, the performance of the system degrades rather than not adopting IRSs. For instance, for $P_{\text{max}}^{\text{UL}}\geq 15$ [dBm], scheme 3 surpasses scheme 2 with a fixed-phase IRSs.}  Moreover, the uplink and downlink sum rates are demonstrated in Fig. \ref{fig:SumRateVsULpower}; the uplink sum-rate increases in $P_{\text{max}}^{\text{UL}}$, and in contrast, the downlink sum-rate degrades by increasing $P_{\text{max}}^{\text{UL}}$.

\subsection{Downlink-Uplink Rate Region Trade-off}
\begin{figure}[!t]
   \centering
   \begin{minipage}[b]{0.45\linewidth}
   \centering
    \psfrag{Scheme 1 E1 AAAAA}[][][0.6]{\!\!\!\!Scheme 1 $E_1$}
    \psfrag{Scemem 2 E1 AAAAA}[][][0.6]{\!\!\!\!\!Scheme 2 $E_1$}
    \psfrag{Scheme 3 E1 AAAAA}[][][0.6]{\!\!\!\!Scheme 3 $E_1$}
    \psfrag{Scheme 1 E2 AAAAA}[][][0.6]{\!\!\!\!Scheme 1 $E_2$}
    \psfrag{Scheme 2 E2 AAAAA}[][][0.6]{\!\!\!\!Scheme 2 $E_2$}
    \psfrag{Scheme 2 E1 BDC AAAAA}[][][0.6]{Scheme 2 $E_1$ BDC}
    \psfrag{Scheme 1 E1 BDC AAAAA}[][][0.6]{Scheme 1 $E_1$ BDC}
    \psfrag{Downlink [bits/s/Hz]}[][][0.76]{Downlink Sum-Rate [bpcu]}
    \psfrag{Uplink [bits/s/Hz]}[][][0.76]{Uplink Sum-Rate [bpcu]}
    \hspace*{-0.67 cm}
    \includegraphics[scale=0.63]{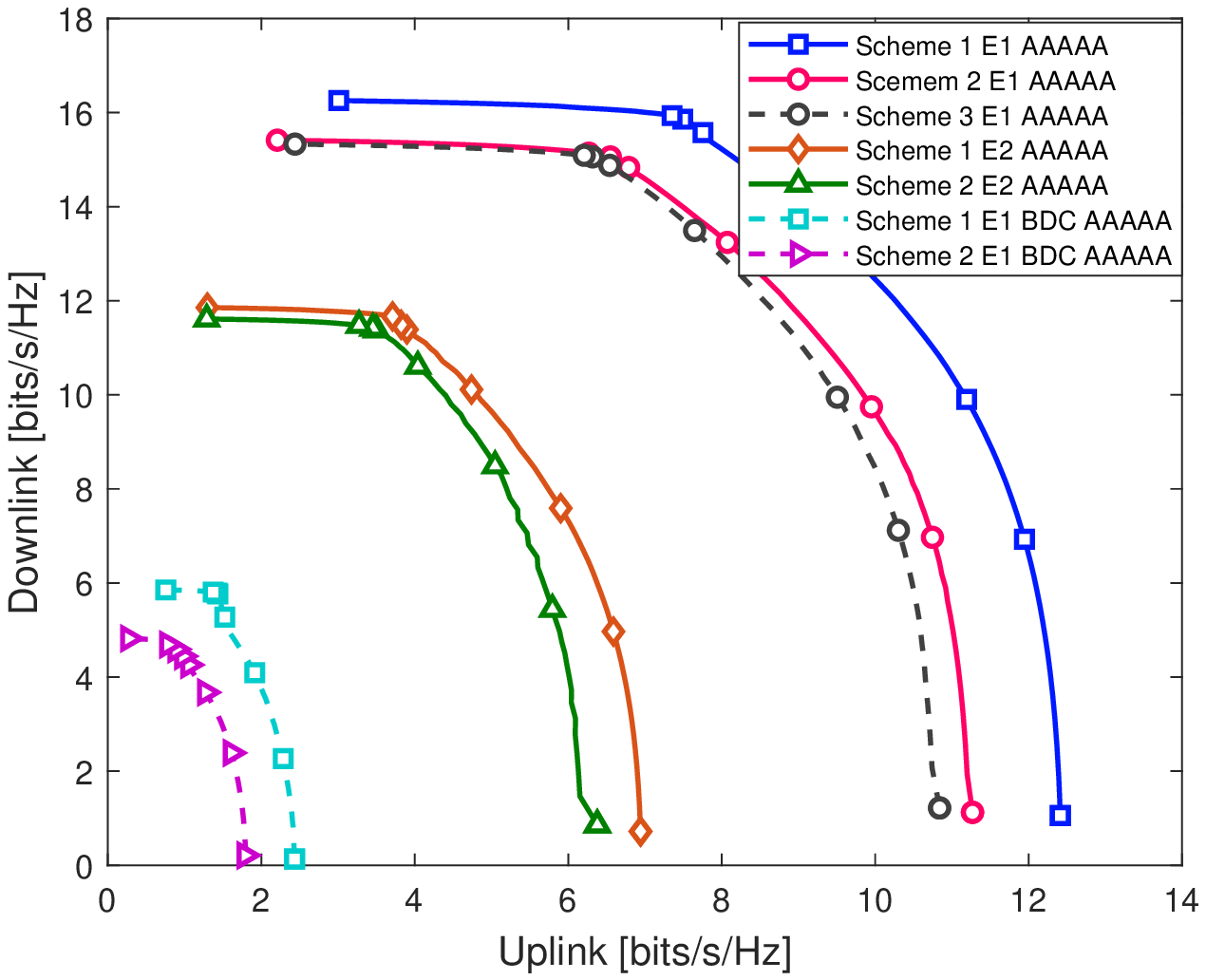}
    \caption{Sum-rate trade-off region of uplink and downlink.}
    \label{fig:Uplink-Downlink-Region}
    \end{minipage}
    \hfill
    \begin{minipage}[b]{0.45\linewidth}
    \centering
    \psfrag{Scheme 1 AAAAAA}[][][0.6]{\!\!\!\!Scheme 1}
    \psfrag{Scheme 2 AAAAAA}[][][0.6]{\!\!\!\!Scheme 2}
    \psfrag{Scheme 1 HI AAAAAA}[][][0.6]{\!Scheme 1 HI}
    \psfrag{Scheme 2 HI AAAAAA}[][][0.6]{\!Scheme 2 HI}
    \psfrag{Scheme 4 AAAAAA}[][][0.6]{\!\!\!\!Scheme 4}
    \psfrag{Scheme 3 AAAAAA}[][][0.6]{\!\!\!\!Scheme 3}
    \psfrag{SSR}[][][0.76]{SWSR [bpcu]}
    \psfrag{CDF}[][][0.76]{CDF}
    \psfrag{AAA}[][][0.5]{$100\%$}
    \hspace*{-1.2 cm}
    \includegraphics[scale=0.63]{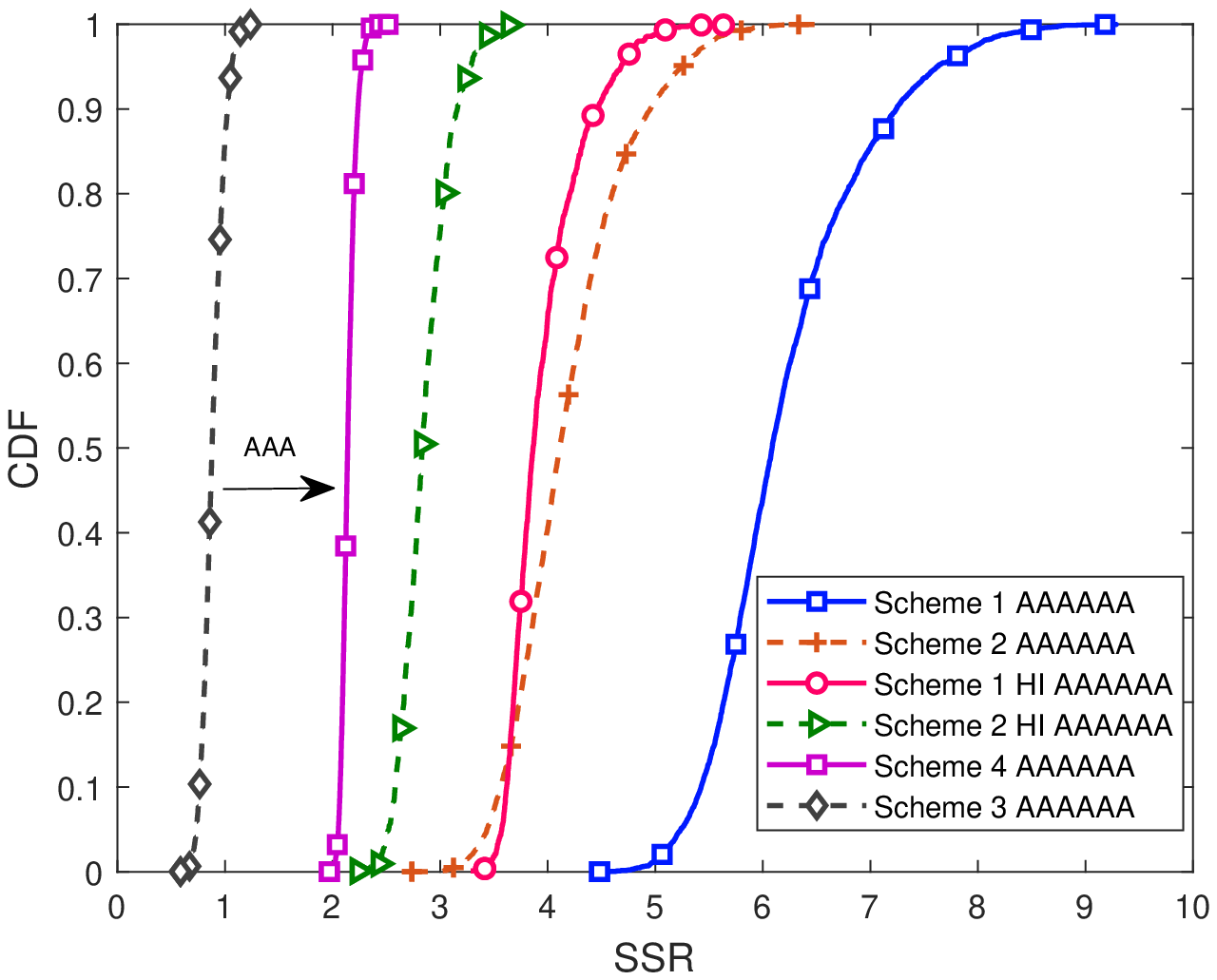}
    \caption{CDF curves for random uplink and downlink user locations.}
    \label{fig:CDF}
    \end{minipage}
\end{figure}

Fig. \ref{fig:Uplink-Downlink-Region} illustrates the sum rates region of the uplink and downlink for various schemes with two sets of parameters  $E_1=\lbrace N_t=4,M=20\rbrace$ and $E_2=\lbrace N_t=2,M=8\rbrace$, and the number of downlink and uplink users are $K=2$ and $L=3$, respectively. \textcolor{black}{In order to observe the impact of adopting multiple IRSs in an environment that  obstacles almost block the direct links}, the performance of the system \textit{without} direct communications links between the BS-to-user pairs are also considered. This case is named blocked direct channels (BDC), e.g. $|\h_k|$ and $ |g_l|$ approach zeros for all users. Generally, it is observed that maximizing the uplink sum-rate degrades the performance of the downlink sum-rate, and vice versa. Also, it is shown that increasing the number of IRSs elements and BS antennas enlarges the region. Finally, for the weak direct channels cases, i.e. BDC case, the region shrinks dramatically. 

\subsection{Cumulative distribution function of SWSR}
The cumulative distribution function (CDF) of the SWSR  for various schemes are presented in Fig. \ref{fig:CDF}. The two IRSs are respectively located at positions $(100,0)$ and $(-100,0)$, the {3} uplink users are located randomly in a circle centred at $(100,14)$ with radius $9$ m and the {2} downlink users are located randomly in a circle centered  at $(-100,-14)$ with radius $9$ m. Also, the size of IRSs and BS antennas are  $M_1=7,M_2=7,N_t=4$, and the direct channels between the BS and the users are blocked. Performance of the three schemes with perfect hardware and non-ideal hardware, i.e. $\kecitDL=\kecitUL=\keciUEDL=\keciUEUL=0.92$, are compared and scheme 4 is also presented. As for scheme 4, it is assumed that the BS performs the conventional  maximum transmission ratio and maximum ratio combining for the DL beamforming and UL combining vectors, respectively, and the users also transmit with full power. Thus, only the phases of IRSs are optimized by use of the gradient ascent approach presented in Algorithm 2. For all the schemes, it is shown that almost the fairness is satisfied and scheme 1 significantly outperforms the other schemes. Also, it is shown that by using the proposed algorithm, the performance of the system is enhanced when a hardware impairment is assumed. Moreover, it can be seen that the scheme 4 obtains about $100\%$ gain compared to that of the scheme 3, and since the direct channels are almost blocked, optimizing only phases of the IRSs even without optimizing other parameters is beneficial.

\subsection{Impact of IRS Locations }
\begin{figure}[!t]
   \centering
   \psfrag{Two IRSs-Case 2 DL AAAA}[][][0.6]{\ Two IRSs-Case 2 DL}
   \psfrag{Single IRS DL AAAAAA}[][][0.6]{\!\!\!\!\!Single IRS DL}
   \psfrag{Two IRSs-Case 1 DL AAA}[][][0.6]{\ \ \ Two IRSs-Case 1 DL}
   \psfrag{Two IRSs-Case 1 UL AAAA}[][][0.6]{\ Two IRSs-Case 1 UL}
   \psfrag{Single IRS UL AAAAAA}[][][0.6]{\!\!\!\!\!Single IRS UL}
   \psfrag{Two IRSs-Case 2 UL AAAA}[][][0.6]{\ Two IRSs-Case 2 UL}
   \psfrag{x=-100(m)}[][][0.52]{{$X_{\text{IRS}}$=-100(m)}}
   \psfrag{x=100(m)}[][][0.52]{{$X_{\text{IRS}}$=100(m)}}
   \psfrag{X}[][][0.76]{$X_{\text{IRS}}$}
   \psfrag{SSR [bits/s/Hz]}[][][0.76]{{Sum-Rate [bpcu]}}
\includegraphics[scale=0.73]{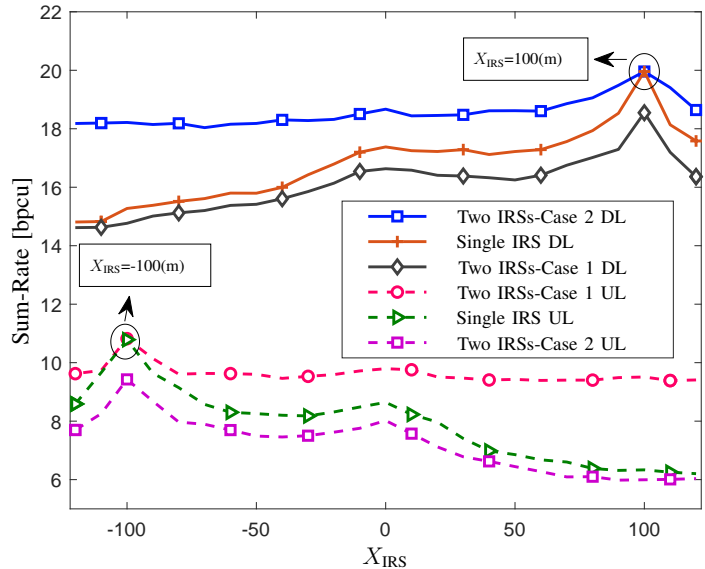}
\caption{Achievable sum rates of UL and DL versus  IRSs location for single IRS and two IRSs.}
\label{fig:IRSlocations}
\end{figure}
Finally, Fig. \ref{fig:IRSlocations} studies the impact of IRSs locations on the uplink and downlink sum-rate, respectively. It is assumed that the locations of $2$ downlink users are set to $(100,24)$ and $(100,15)$, which is called downlink users \textit{zone}, and $3$ uplink users are positioned at as $(-100,24)$, $(-100,15)$, and $(-100,25)$, respectively, which is also called uplink users \textit{zone}. To analyze the impact of IRSs locations in two-dimensional constrained space, two scenarios of  \textit{single} IRS and \textit{two} IRSs are considered. For the single IRS, it is assumed that the IRS has  24 elements, and is located at $(X_{\text{IRS}},20)$ where $X_{\text{IRS}} \in [-120,120]$, i.e. its location can change on the line from  $(-120,20)$ to $(120,20)$. For the two IRSs, to have a fair comparison, it is assumed that size of each IRS is 12 and two cases are considered; 1) IRS 1 moves from $(120,20)$ to $(-120,20)$ while IRS 2 is fixed at $(-120,20)$, 2) IRS 2 moves from $(-120,20)$ to $(120,20)$, while IRS 1 is fixed at $(120,20)$. Therefore, the uplink sum-rate and downlink sum-rate of single and two IRSs scenarios versus the location of varying IRS is depicted in Fig. \ref{fig:IRSlocations}. It can be observed that for the two IRSs cases, \textcolor{black}{when the location of the IRS 2 changes from the uplink users zone to the downlink users zone,} the downlink sum-rate increases and achieves its maximum at location $(100,20)$ where the IRS 1 is placed. At this point, as depicted in Fig. \ref{fig:IRSlocations}, the downlink sum-rate meets the performance of the single IRS scenario. Also, for this case, the uplink sum-rate decreases as the IRS 2 goes farther from the uplink users zone. Moreover, for the case 1, by moving the IRS 1 from downlink users zone to the uplink users zone, the uplink sum-rate increases and achieves its largest value at $(-100,20)$, and it is clear that the downlink sum-rate decreases, as well. Also, in case 1, the performance of two IRSs is the same as the single IRS at $(-100,20)$. As a result, in the two IRSs cases, due to the presence of one IRS in either uplink or downlink users zone, \textcolor{black}{we can conclude that despite the double number of elements of a single IRS compared with the number of elements for each IRS in cases 1 and 2, the two IRSs scenarios provide more stable and acceptable results for both uplink sum-rate and downlink sum-rate.} 
\section{Conclusions}
\label{sec:concolusion}
In this paper, we investigated the effect of deploying multiple IRSs in a FD multi-user communication system. Specifically, our study focused on joint optimization for IRSs phase shift matrices, the beamformer and combining vectors at the BS, and the transmitted power of the uplink users. \textcolor{black}{Also, it is assumed that there is hardware impairment at users, transmitter, and receiver of the BS.}
The SWSR maximization problem subject to the maximum power constraints at the BS and the uplink users was considered, and an iterative algorithm was proposed. Due to the non-convexity of the optimization problem, we tackled the problem by utilizing the alternating optimization method wherein the WMMSE approach was also used. We firstly transformed the optimization problem into several convex sub-problems and handled them by applying the Lagrangian multiplier method to analytically derive the optimal solutions. Moreover, in our proposed algorithm, the optimized phase shifts of IRSs were obtained via a gradient ascent-based method by solving an unconstrained equivalent problem. The complexity of the overall proposed algorithm was discussed, and its convergence was verified through numerical results. Finally, the effects of the transmission power of the BS and the uplink users, the size and the location of IRSs were discussed and compared for various topologies to clarify performance enhancement of the proposed algorithm. \textcolor{black}{Moreover, by using multiple optimized IRSs, SWSR is improved when the users and the BS have hardware impairment.} It is concluded that utilizing multiple distributed IRSs in a FD scenario is more beneficial than using a centralized single IRS. 

\bibliography{ref}
\bibliographystyle{IEEEtran}

\end{document}